\documentclass[journal,final]{IEEEtran}
\usepackage{zref-base}
\usepackage{atveryend}
\makeatletter
\AfterLastShipout{%
  \if@filesw   
    \begingroup
      \zref@setcurrent{default}{\the\value{equation}}%
      \zref@wrapper@immediate{%
        \zref@labelbyprops{eq-last}{default}%
      }%
    \endgroup
  \fi
}
\makeatother

\usepackage{cite}

\ifCLASSINFOpdf
  \usepackage[pdftex]{graphicx}
\else
  \usepackage[dvips]{graphicx}
\fi
\usepackage{amsmath}
\usepackage{algpseudocode}
\newsavebox{\ieeealgbox}
\newenvironment{boxedalgorithmic}
  {\begin{lrbox}{\ieeealgbox}
   \begin{minipage}{\dimexpr\columnwidth-2\fboxsep-2\fboxrule}
   \begin{algorithmic}[1]}
  {\end{algorithmic}
   \end{minipage}
   \end{lrbox}\noindent\fbox{\usebox{\ieeealgbox}}}

\usepackage{array}

\ifCLASSOPTIONcompsoc
  \usepackage[caption=false,font=normalsize,labelfont=sf,textfont=sf]{subfig}
\else
  \usepackage[caption=false,font=footnotesize]{subfig}
\fi

\usepackage{url}

\usepackage{hyperref}
\usepackage{psfrag}
\usepackage{dsfont,times,import}
\usepackage{amssymb}
\usepackage{color}
\usepackage{makecell}
\usepackage{amsthm}

\newtheorem{theorem}{Theorem}

\newtheorem{assumption}{}

\DeclareMathOperator*{\argmax}{arg\,max}

\begin{document}

\title{Algorithmic Bidding for Virtual Trading\\ in Electricity Markets}

\author{Sevi~Baltaoglu,~\IEEEmembership{}
            Lang~Tong,~\IEEEmembership{Fellow,~IEEE,}
            and~Qing~Zhao,~\IEEEmembership{Fellow,~IEEE}
\thanks{Sevi Baltaoglu, Lang Tong, and Qing Zhao are with the School of Electrical and Computer Engineering, Cornell University, Ithaca, NY, 14850 USA e-mail: {\tt \{msb372,lt35,qz16\}@cornell.edu}.}%
\thanks{This work was supported in part by the National Science Foundation under Award 1809830 and 1816397.}%
\thanks{Part of the work was presented at Conference on Neural Information Processing Systems (NIPS) 2017.}}

\markboth{Journal of \LaTeX\ Class Files,~Vol.~14, No.~8, August~2015}%
{Shell \MakeLowercase{\textit{et al.}}: Bare Demo of IEEEtran.cls for IEEE Journals}


\maketitle

\begin{abstract}
We consider the problem of optimal bidding for virtual trading in two-settlement electricity markets. A virtual trader aims to arbitrage on the differences between day-ahead and real-time market prices; both prices, however, are random and unknown to market participants. An online learning algorithm is proposed to maximize the cumulative payoff over a finite number of trading sessions by allocating the trader's budget among his bids for $K$ options in each session. It is shown that the expected payoff of the proposed algorithm converges, with an almost optimal convergence rate, to the expected payoff of the global optimal corresponding to the case when the underlying price distribution is known. The proposed algorithm is also generalized for trading strategies with a risk measure. By using both cumulative payoff and Sharpe ratio as performance metrics, evaluations were performed based on historical data spanning ten year period of NYISO and PJM markets. It was shown that the proposed strategy outperforms standard benchmarks and the S\&P 500 index over the same period.  
\end{abstract}


%
\IEEEpeerreviewmaketitle

\section{Introduction}
\IEEEPARstart{T}{he} wholesale electricity market in the United States consists of a day-ahead (DA) and a real-time (RT) markets. Market participants submit their bids to buy (and offers to sell) electricity to the DA market approximately one day ahead of time. The bids and offers cleared in the DA market are financially binding. The market clearing process sets the DA prices for each hour of the day and at each location of the network. 

In the RT market, the load (thus the generation) may not match to the cleared amount in the DA market, and the RT prices of electricity may also be different from their DA counterparts due to a variety of reasons, including the unexpected levels of demand and supply, unplanned outages, unpredictable weather conditions\cite{Lietal:15}, and possibilities of market participants exercising market power\cite{PJM:15}.

Price discrepancies between the DA and RT markets represent a form of market inefficiency. To promote {\em price convergence} between the two markets, in early 2000s, virtual trading was introduced in the U.S. electricity markets. Virtual trading is a financial mechanism that allows market participants and external financial entities to arbitrage on the differences between DA and RT prices. Currently, cleared virtual transactions represent a significant fraction of total energy trade. In 2013, the cleared virtual transactions in the five major electricity markets was 13\% \footnote{This number goes up to 38\% with the inclusion of up-to-congestion transactions of PJM.} of the total load\cite{Parsonsetal:15}. 

Empirical and theoretical studies have shown that increased competition due to virtual trading results in price convergence, thus improving market efficiency\cite{PJM:15, Lietal:15, Hogan:16, Tangetal:16, Jha&Wolak:15, Saravia:03, Guleretal:10, Wooetal:15}. Particularly, it has been argued in\cite{Tangetal:16} that a virtual trader makes profit if and only if his participation drives the DA-RT price difference toward zero. Hence, to reach the socially optimal dispatch level, it is important that the virtual traders bid optimally. However, the DA and RT wholesale prices are random due to uncertainties in demand, supply, and operation conditions. Therefore, in order to learn the optimal trading strategy, a virtual trader needs to update his belief using all the new information, which allows him to adapt his bid accordingly each day. 

\subsection{Main Contributions}

In an electricity market, there are potentially thousands of trading options. Due to system congestion and losses, electricity prices vary in time and across locations. The goal of this work is to develop an {\em online learning} approach to virtual trading where the trader, who is constrained by a certain budget, aims to determine profitable trading options and distribute his budget among them. By online learning we mean that bids are constructed sequentially and adaptively based on the new information available. In particular, we consider the objective of maximizing the expected total payoff as well as one that involves a mean-variance type of risk measure. 

The main contribution of this work is a polynomial-time online learning algorithm that maximizes the expected cumulative return over a finite $T$ trading horizon. This result is also generalized to an objective based on a form of mean-variance risk measure.  Note that obtaining the optimal bidding strategy with known joint DA and RT price distribution is itself nontrivial due to the non-convexity of the problem.  Our result provides an algorithmic bidding strategy that converges to the global optimal bidding strategy.  We show further that the rate of convergence achieves the lower bound of convergence of all such algorithms up to a $\sqrt{\log(T)}$ factor. To establish the order-optimality of convergence, the regret measure is used. In particular, the regret of an online learning algorithm is the difference between the expected total T-period payoff of that algorithm and that of the global optimal bidding strategy when the underlying probability model of DA and RT prices is known. 

A significant part of this work is to evaluate the performance of the proposed strategy empirically using historical data spanning the time period between 2006 and 2016 of NYISO and PJM markets. Extensive empirical analysis show that the proposed strategy consistently outperforms benchmark heuristic methods, derived from other machine learning approaches, and achieves significant profit. It is worth noting that our empirical results also show that PJM and NYISO wholesale electricity markets are both profitable although PJM market presents better opportunities to traders compared to NYISO. 

\subsection{Related Work}

Relevant literature falls into two categories. The one that is more directly related to our work focuses on developing online learning algorithms to similar problems in the machine learning literature. The second one focuses on understanding the effects of virtual trading on the two-settlement electricity market.
 
Within the machine learning literature, most relevant work is algorithmic bidding in online advertising markets\cite{Weedetal:16}. The authors of\cite{Weedetal:16} consider the online bidding problem without a budget constraint where the bidder observes market prices only if his bid is accepted at that period. In contrast, our problem falls into the category of stochastic \emph{experts} problem because historical market prices are observable regardless of the bid of the virtual trader. The analysis in \cite{Weedetal:16} on the lower bound of {\em regret} does not hold for the problem considered here. Furthermore, their algorithm cannot be used here due to the budget constraint. Our problem is a special case of the setting studied by Kleinberg and Slivkins\cite{Kleinberg&Slivkins:10}. Unfortunately, the computational complexity of the algorithm in\cite{Kleinberg&Slivkins:10} grows exponentially with the number of options. Therefore, it becomes intractable in practice. Also, the regret lower bound in\cite{Kleinberg&Slivkins:10} doesn't provide a bound for our problem with a specific payoff. 

Indirectly related to this work are works that analyze the impact of virtual transactions on the overall market efficiency. Theoretical analysis on the impact of virtual trading was conducted in \cite{Tangetal:16} and \cite{Matheretal:17} from a game theoretic perspective. Under a single trading location model, these papers analyzed the Nash equilibrium (NE) behavior of virtual traders who have their fixed individual beliefs about the market. Tang et al. \cite{Tangetal:16} showed that, under NE, if the belief of virtual traders is correct on average, the price difference between DA and RT converges to zero as the number of virtual traders increases. Mather, Bitar, and Poolla \cite{Matheretal:17} presented a simple learning strategy that guarantee convergence to the NE. However, convergence to NE doesn't guarantee price convergence. Different from the problem of learning the NE in a game theoretic environment with fixed beliefs\cite{Matheretal:17}, we study the online learning problem of a virtual trader who updates his belief each day using new observations of DA and RT prices in order to converge to the optimal trading strategy. Furthermore, we require that not only the bidding policy converges to the optimal policy but also the convergence rate is order-optimal. 

Among empirical studies, \cite{Lietal:15} and \cite{Jha&Wolak:15} are the most relevant to our work. Both evaluate market efficiency before and after virtual trading was introduced in the electricity markets. More specifically, in \cite{Lietal:15}, a chance constraint portfolio selection problem was solved by estimating the distribution of DA-RT price difference, modeled as Gaussian mixture hidden Markov model, to determine the trading strategy, whereas, in \cite{Jha&Wolak:15}, hypothesis testing is used to determine the existence of a profitable trading strategy at each location. Empirical analysis using CAISO data shows that virtual trading increases efficiency but the market is still inefficient. Some of the other interesting empirical studies on the impact of virtual trading are \cite{Borensteinetal:08, Saravia:03, Birgeetal:17}, and \cite{Guleretal:10}. 

\section{Virtual Trading in Electricity Markets}

\subsection{Virtual Transactions in the Two-Settlement Market System}
 
A virtual transaction on any given day (session) involves transactions in the DA and RT markets for power at a particular location and in a particular hour.  Herein we refer to each location-hour pair with which a transaction is associated as a trading option. Typically, two types of virtual transactions are allowed in the US wholesale electricity markets: (i) virtual demand bid and (ii) virtual supply bid. A virtual demand bid is a bid to buy energy in the DA market with an obligation to sell back exactly the same amount in the RT market. A virtual supply bid is a bid to sell energy in the DA market with an obligation to buy back exactly the same amount in the RT market. 

The DA market takes place one day ahead of the actual power delivery.  In the DA market, the independent system operator (ISO) receives bids from (actual) generators and load serving entities as well as virtual bidders. After the DA market closes on day $t-1$, the bids in the DA market are processed by the ISO via a security constrained economic dispatch that accepts a subset of virtual bids and determines the amount of power to generate for each generator and the associated DA prices.  
	
The RT market takes place at the time of actual power delivery on day $t$.  The ISO adjusts the dispatch level according to the actual system operating conditions and compute the RT prices.  The virtual bids that are accepted in the DA market are settled in the RT market, and a virtual bidder with an accepted bid is paid at the difference of the DA and RT prices.  We give next a more precise mathematical description of the settlement process.

\subsection{A Mathematical Model of Virtual Trading}
\label{sec:virtual_trading_model}

Recall that a trading option is defined by a pair of a location and a particular time of power delivery. A location can be a bus of the transmission grid or a trading zone.  The time of power delivery is a specific hour in a 24 hour trading horizon.   

Let $\lambda_{t,k}$ and $\pi_{t,k}$ be the DA and RT prices (in \$/MWh) of option $k$ on day $t$, respectively. Let $x_{t,k}$ be a virtual bid (in \$/MWh) for option $k$ on day $t$. A {\em virtual demand bid} is a bid to buy a unit quantity of electricity at a particular location and hour in the DA market with the obligation to sell the same amount at the same location and hour in the RT market. The demand bid $x_{t,k}$ is cleared if the bid price $x_{t,k}$ is higher than or equal to the DA price $\lambda_{t,k}$, \textit{i.e.}, $x_{t,k} \geq \lambda_{t,k}$. For the accepted bid, the payoff is the difference between the RT and DA prices of that option, \textit{i.e.},
\begin{displaymath}
(\pi_{t,k} - \lambda_{t,k}) \mathds{1}\{x_{t,k} \geq \lambda_{t,k}\}
\end{displaymath}
where $\mathds{1}\{\cdot\}$ is the indicator function that is one if its argument is true and zero otherwise. 

Similarly, a {\em virtual supply bid} is an offer to sell electricity in the DA market with the obligation to buy back in the RT market. The supply bid $x_{t,k}$ is cleared if the bid price $x_{t,k}$ is lower than or equal to the market clearing price $\lambda_{t,k}$, \textit{i.e.}, $x_{t,k} \leq \lambda_{t,k}$. For the accepted bid, the payoff is given by
\begin{displaymath}
(\lambda_{t,k}- \pi_{t,k}) \mathds{1}\{x_{t,k} \leq \lambda_{t,k}\}.
\end{displaymath}

The payoff for the two types of bids can be expressed by a single expression through a simple translation.  To this end, we assume that DA prices are bounded with known upper/lower bounds, \textit{i.e.} $l_\lambda< \lambda_{t,k} < u_\lambda$. Then, regardless of the type of bids, the payoff obtained from option $k$ on day $t$ can be written as: 
\begin{displaymath}
(\pi_{t,k}' - \lambda_{t,k}') \mathds{1}\{x_{t,k}' \geq \lambda_{t,k}'\}
\end{displaymath}
where, for a virtual demand bid, $x_{t,k}'=x_{t,k}-l_\lambda$, $\lambda_{t,k}'=\lambda_{t,k}-l_\lambda$, and $\pi_{t,k}'=\pi_{t,k}-l_\lambda$; and, for a virtual supply bid, $x_{t,k}'=u_\lambda-x_{t,k}$, $\lambda_{t,k}'=u_\lambda-\lambda_{t,k}$, and $\pi_{t,k}'=u_\lambda-\pi_{t,k}$. In this case, observe that $x_{t,k}' =0$ is equivalent to not bidding for option $k$ on day $t$. 

For notational convenience, hereafter we use $\lambda_{t,k}$, $\pi_{t,k}$, and $x_{t,k}$ instead of $\lambda_{t,k}'$, $\pi_{t,k}'$, and $x_{t,k}'$ to represent the translated price and bid variables. The accumulative return for a T-period trading horizon for a given bid $x_{t,k}$ sequence and DA/RT prices, irrespective the type of bids, is given by
\begin{equation} 
\label{eq:CumulativeReturn}
\sum_{t=1}^T (\pi_{t,k} - \lambda_{t,k}) \mathds{1}\{x_{t,k} \geq \lambda_{t,k}\}.
\end{equation}

Next, we study the problem of a virtual trader who considers to bid on $K$ options and aims to determine the optimal value of $x_{t,k}$ for each $k \in \{1,...,K\}$ under a budget constraint. Note that a trader can submit multiple bids for the same option, including demand and supply bids, simultaneously.

\section{Online Learning Approach to Virtual Trading}

In this section, we develop an algorithmic bidding strategy aimed at maximizing expected payoff by allocating a fixed budget among $K$ options without assuming the knowledge of underlying joint distribution of the DA and RT prices. 

An outline of our proposed approach is in order. Since the expected payoff cannot be calculated analytically due to the unknown distribution, we consider the maximization of the sample mean payoff, which is equivalent to an empirical risk minimization (ERM) problem\cite{Vapnik:92}. For fixed trading horizon $T$, solving this ERM amounts to solving a multiple-choice knapsack problem\cite{Kellereretal:04}, which is NP hard.  We propose a polynomial-time approximation algorithm, referred to as dynamic programming on discrete set (DPDS) that converges to the optimal bidding strategy as the number of trading sessions increases. More importantly, the DPDS is order optimal in terms of its rate of convergence as shown in Sec.\ref{sec:convergence}. We also extend this algorithm to deal with the objective of optimizing a variant of mean-variance measure.

\subsection{Problem Formulation}

Let $\lambda_{t}=[\lambda_{t,1},...,\lambda_{t,K}]^\intercal$ and $\pi_{t}=[\pi_{t,1},...,\pi_{t,K}]^\intercal$ be the vector of DA and RT prices on day $t$, respectively. Similarly, let $x_t = [x_{t,1},...,x_{t,K}]^\intercal$ be the vector of bids for day $t$. At the end of each day, the DA and RT prices of all options are observed. Therefore, before choosing the bid for day $t$, all the information the virtual trader has is a vector\footnote{In practice, the bid for day $t$ needs to be chosen before observing the full vector of RT prices of day $t-1$. However, in that case, $I_{t-1}=\{x_i,\lambda_i,\pi_i\}_{i=1}^{t-2}$ can be used instead without loss of generality.} $I_{t-1}$ containing his observation and decision history $\{x_i,\lambda_i,\pi_i\}_{i=1}^{t-1}$. Consequently, a bidding policy $\mu$ is defined as a sequence of decision rules, \textit{i.e.}, $\mu = (\mu_0,\mu_1...,\mu_{T-1})$, such that, at time $t-1$, $\mu_{t-1}$ maps the information history $I_{t-1}$ to the bid $x_t$ of day $t$.

The objective is to determine a bidding policy $\mu$ that maximizes the expected cumulative payoff over T days subject to a budget constraint for each individual day. From (\ref{eq:CumulativeReturn}), the optimization problem can be written as
\begin{align}
\label{eq:optproblem}
& \underset{\mu}{\text{maximize}}
& &\mathbb{E}\left( \sum_{t=1}^T (\pi_{t} - \lambda_{t})^\intercal \mathds{1}\{x_{t}^\mu \geq \lambda_{t}\} \right)  \\
& \text{subject to}
& & \|x_{t}^\mu\|_1 \leq B, \qquad \forall t=1,...,T,\notag \\
& & & x_{t}^\mu \geq 0, \quad \qquad \forall t=1,...,T,\notag
\end{align}
where $x_t^\mu$ is the (translated) bid determined by policy $\mu$, $\mathds{1}\{x_{t}^\mu \geq \lambda_{t}\}$ the vector of indicator functions with the $k$th entry corresponding to $\mathds{1}\{x_{t,k}^\mu \geq \lambda_{t,k}\}$, and $B$ the auction budget\footnote{This budget provides an upper bound to DA market spending in the case of demand bids only and non-negative DA prices. However, it becomes artificial with the inclusion of supply bids. In the general setting, the budget constraint restricts the number of options to bid and leads to the determination of bid values that provides the best payoff per unit of bid.} of the virtual trader. The expectation is taken with respect to randomness in  $\{\pi_t, \lambda_t\}_{t=1}^T$ and policy $\mu$.

The joint distribution of the DA and RT prices is unknown to the virtual trader. Hence, it is not possible to solve the optimization problem analytically. Even if the joint distribution was known, the above optimization can be non-convex, and obtaining global optimal policy is nontrivial. Instead, virtual trader uses his observation history to obtain the optimal bid.

\subsection{An ERM approach}

Because past DA and RT prices are observable, one can calculate the (empirical) average payoff that could have been obtained up to the current day by a fixed bid $x \in \mathcal{F}$ where $\mathcal{F} = \{x \in \Re^K:x\geq 0, \|x\|_1 \leq B\}$ is the feasible set of bids. Specifically, the average payoff $\bar{r}_{t,k}(x_k)$ from option $k$ with fixed bid $x_k$ in $t$ trading sessions is 
\begin{displaymath}
\bar{r}_{t,k}(x_k)= \frac{1}{t} \sum_{i=1}^t (\pi_{i,k}-\lambda_{i,k})\mathds{1}{\{x_k \geq \lambda_{i,k}\}}.
\end{displaymath}

Note that $\bar{r}_{t,k}(x_k)$ is a piece-wise constant function with at most $t$ breakpoints, each corresponding to a new DA price observed in the past $t$ periods. Let the vector of order statistics of the observed DA prices $\{\lambda_{i,k}\}_{i=1}^t$ be $\lambda^{(t,k)}$ as illustrated in Fig.~\ref{fig:averagepayoff3} (for $t=4$). Let $r^{(t,k)}$ be the associated vector of average payoffs where $r^{(t,k)}_j$ is the average payoff $\bar{r}_{t,k}(\lambda^{(t,k)}_j)$ for fixed bid $\lambda^{(t,k)}_j$. Then, $\bar{r}_{t,k}(x_k)$ can be expressed by the pair $\left(\lambda^{(t,k)} , r^{(t,k)} \right)$ as shown in Fig.~\ref{fig:averagepayoff3}.

\begin{figure}[!t]
\centering
\def\svgwidth{0.85\linewidth}
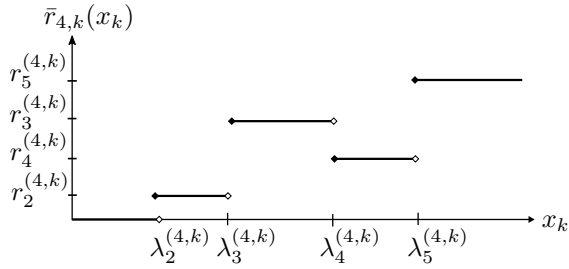
\caption{Example of a piece-wise constant average payoff function of option $k$ when $t=4$.}
\label{fig:averagepayoff3}
\end{figure}

The overall average payoff function $\bar{r}_t(x)$ is the sum of average payoff functions of individual options. Now, let's consider the maximization of $\bar{r}_t(x)$,
\begin{equation}
\label{eq:avg_opt}
\max_{x \in \mathcal{F}} \bar{r}_t(x) =\max_{x \in \mathcal{F}} \sum_{k=1}^K \bar{r}_{t,k}(x_k).
\end{equation}
The above ERM problem in (\ref{eq:avg_opt}), unfortunately, is NP-hard. This result can be obtained by showing that, due to the piece-wise constant structure of the average payoff function, solving (\ref{eq:avg_opt}) is equivalent to solving a multiple-choice knapsack problem (MCKP), which is known to be NP-hard (see Appendix~\ref{appendix:np_hardness}).

\subsection{DPDS: A Polynomial-Time Online Learning Algorithm}

We now derive a polynomial-time algorithm referred to as dynamic programming on discrete set (DPDS). The idea behind DPDS is to discretize the feasible set using intervals of equal length and optimize the average payoff on this new discrete set via a dynamic program. 

Let $\alpha_t$ be an integer sequence increasing with $t$, and $\mathcal{D}_t = \{0,B/\alpha_t, 2B/\alpha_t,...,B\}$ is a sequence of equally placed grid points in $[0, B]$ with increasing density wtih $t$. Then, the new discrete set is given as
$\mathcal{F}_t= \{x \in \mathcal{F}: x_k \in  \mathcal{D}_t,  \forall k \in \{1,...,K\}\}$. Our goal is to optimize $\bar{r}_t(.)$ on the new set $\mathcal{F}_t$ rather than $\mathcal{F}$, \textit{i.e.},
\begin{equation}
\label{opt:dp_opt}
\max_{x_{t+1} \in \mathcal{F}_t} \bar{r}_t(x_{t+1})= \max_{x_{t+1} \in \mathcal{F}_t} \sum_{k=1}^K \bar{r}_{t,k}(x_{t+1,k}).
\end{equation}

Observe that, for fixed $t$, this can be written as a multistage decision problem with $K$ stages as follows: the state of stage $k$ is the remaining budget $b_k \in D_t$, and  the action (decision) of stage $k$ is the bid value $x_{t+1,k} \in A_{t,b_k}$ of option $k$ where $A_{t,b_k}=\{x_k \in D_t: x_k \leq b_k\}$. In this case, $D_t$ is the state space, $A_{t,b_k}$ the action space of stage $k$, and $\bar{r}_{t,k}(x_{t+1,k})$ the payoff (reward) of stage $k$ for taking action $x_{t+1,k}$. 

Now, we define the maximum payoff one can collect in state $b$ over the remaining $n$ stages as $V_n(b)$.  Then, the Bellman equation can be used to solve for $V_K(B)$ which gives the optimal solution to (\ref{opt:dp_opt}). This type of dynamic programming approach has been used to solve 0-1 Knapsack problems including MCKP\cite{Dudzinski&Walukiewicz:87}. However, that approach results in pseudo-polynomial computational complexity in the case of 0-1 Knapsack problems. The discretization of the feasible set with equal interval length reduces the computational complexity to polynomial time. More specifically, the total computational complexity of DPDS  is $O(K\max(t,\alpha_t^2))$ at each day $t$ t (see
Appendix~\ref{appendix:comp_complexity} for detailed analysis). 

\subsection{Risk-Averse Learning}

Maximizing expected profit is not necessarily a prudent strategy in algorithmic bidding. Often the risk of a particular strategy needs to be taken into account. A commonly used metric to measure the effectiveness of a strategy is the Sharpe ratio, which is the ratio of the expected return and the standard deviation of the return. In essence, this requires a trade-off between maximizing the expected return and minimizing the variance.

In order to distribute the budget among the options with high payoff and low variance, we extend DPDS algorithm to the optimization of sum of sample mean-variance of all options (a variant of the well known mean-variance strategy \cite{Markowitz:52}). Let $\nu_{i,k}$ denotes $(\pi_{i,k}-\lambda_{i,k})$. The sample mean-variance function for option $k$ can be written as
\begin{align}
\label{eq:sample_mean_var}
\bar{r}_{t,k}^{(\rho)}(x_k) & = \bar{r}_{t,k}(x_{k}) \notag\\
 &\quad - \frac{\rho}{t-1}\sum_{i=1}^t\left(\nu_{i,k}\mathds{1}\{x_{k} \geq \lambda_{i,k}\}-\bar{r}_{t,k}(x_{k})\right)^2 \notag\\
&  = \bar{r}_{t,k}(x_{k}) +\rho \frac{t}{t-1} \bar{r}_{t,k}(x_{k})^2 \notag\\
& \quad -\rho\frac{t}{t-1}\left(\frac{1}{t}\sum_{i=1}^t\nu_{i,k}^2\mathds{1}\{x_{k} \geq \lambda_{i,k}\}\right). 
\end{align}

In the last equality, observe that $\bar{r}_{t,k}^{(\rho)}(x_k)$ is also a piece-wise constant function with the same breakpoints as $\bar{r}_{t,k}(x_{k})$. So, the extension of DPDS follows from using $\bar{r}_{t,k}^{(\rho)}(x_t)$ instead of  $\bar{r}_{t,k}(x_{k})$ while solving the Bellman equation. The value of $\bar{r}_{t,k}^{(\rho)}(x_k)$ can be obtained by additionaly updating the value of the last term in (\ref{eq:sample_mean_var}). This update is similar to $\bar{r}_{t,k}(x_{k})$ update. (See the algorithm pseudo-code given in Fig.~\ref{alg:dpds}.)

\subsection{Order Optimality of DPDS}
\label{sec:convergence}

We present a performance analysis of DPDS in this section. Our results are of two types. First is to show that DPDS converges to the globally optimal bidding strategy as the trading horizon $T \rightarrow \infty$.  The second is to show that the rate of convergence of DPDS is order-optimal up to a $\sqrt{\log(T)}$ factor. This result shows that DPDS has a strong convergence property over finite trading  horizons.
  
Observe that $\bar{r}_{t,k}^{(\rho)}(x_k) = \bar{r}_{t,k}(x_k)$ when $\rho =0$. Here, we present our analysis for the sum of mean-variance case (for any choice of $\rho\geq0$) that includes the expected return ($\rho =0$) as a special case.

For performance analysis, it is necessary to make several assumptions.  These assumptions do not limit the implementation of the algorithm; they are necessary to make the performance guarantee of DPDS precise. The assumptions, \ref{assumption:iid}, \ref{assumption:bounded}, and \ref{assumption:lipschitz} that are used for performance analysis are given below.

\begin{assumption} 
\label{assumption:iid}
The DA and RT prices $(\lambda_t,\pi_t)$ are drawn independently\footnote{For a similar assumption, see Jha and Wolak\cite{Jha&Wolak:15}, who showed that one cannot reject the hypothesis that the autocorrelation matrices of DA-RT price differences beyond first leg are zero. Hence, the assumption is reasonable due to prices of day $t-1$ being unobservable before bidding for day $t$ in reality.} and identically\footnote{This implies that the DA price is independent of $x_t$, which is reasonable for any market where an individual has negligible impact on the market price.} over time $t$ from an unknown joint distribution $f(\lambda_t,\pi_t)$. 
\end{assumption}

\begin{assumption} 
\label{assumption:bounded}
The payoff resulting from bidding on any node $k \in \{1,...,K\}$ is a bounded random variable with support in $[l, u]$ for any $x \in \mathcal{F}$, \textit{i.e.} $l\leq (\pi_{t,k}-\lambda_{t,k})\mathds{1}\{x_{k} \geq \lambda_{t,k}\} \leq u$.
\end{assumption}

Recall that $\nu_{t,k}$ denotes $(\pi_{t,k}-\lambda_{t,k})$. Define the expected payoff at day $t$ of node $k$ given the bid $x_{t,k}$ as
\begin{displaymath}
r_k(x_{t,k}) = \mathbb{E}(\nu_{t,k}\mathds{1}\{x_{t,k} \geq \lambda_{t,k}\}|x_{t,k}),
\end{displaymath}
and the variance of the payoff of node $k$ given the bid $x_{t,k}$ as
\begin{displaymath}
v_k(x_{t,k}) = \mathbb{E}\left((\nu_{t,k}\mathds{1}\{x_{t,k} \geq \lambda_{t,k}\}-r_k(x_{t,k}))^2 |x_{t,k}\right).
\end{displaymath}
Then, the sum of mean-variance of all nodes will be given by
\begin{displaymath}
r^{  (\rho) } (x_t) = \sum_{k=1}^K (r_k(x_{t,k})  - \rho v_k(x_{t,k}) ).
\end{displaymath}

\begin{assumption} 
\label{assumption:lipschitz}
$r^{(\rho)}(.)$ is Lipschitz continuous on $\mathcal{F}$ with p-norm and Lipschitz constant $L$.
\end{assumption}

Observe that if DA and RT prices have a bounded support and the distribution $f(\lambda_t,\pi_t)$ is uniformly continuous and uniformly bounded on the union of that support and the feasible set $\mathcal{F}$, then assumptions \ref{assumption:bounded} and  \ref{assumption:lipschitz} are satisfied.

For $\rho \geq 0$, the problem of the virtual trader is to find a bidding policy $\mu$ such that
\begin{equation}
\label{eq:genoptprob}
\max_{\mu:\text{ $x_t^\mu \in \mathcal{F}$ $\forall$ $t$}}  \mathbb{E}\left(\sum_{t=1}^T r^{(\rho)}(x_t^\mu)\right),
\end{equation}
which is equivalent to (\ref{eq:optproblem}) when $\rho=0$. Due to \ref{assumption:iid}, optimal solution to (\ref{eq:genoptprob}) under known distribution of $(\pi_{t},\lambda_{t})$ does not depend on $t$ and is given by
\begin{displaymath}
x^* = \argmax_{x \in \mathcal{F}} r^{  (\rho) } (x). 
\end{displaymath}

Following the online machine learning literature, we measure the performance of any bidding policy $\mu$ by its regret $\mathcal{R}_T^\mu(f)$, defined by the difference between the total expected payoff of policy $\mu$ and that of the optimal solution $x^*$, \textit{i.e.},
\begin{displaymath}
\mathcal{R}_T^\mu(f) = \sum_{t=1}^T\mathbb{E}\left(r^{(\rho)}(x^*) - r^{(\rho)}(x^\mu_t)\right).
\end{displaymath}
By definition, the regret is monotonically increasing for any policy $\mu$ and grows linearly with $T$ for the worst possible $\mu$. Since we define optimality as maximizing the expected payoff, observe that a policy $\mu$ converges to the optimal solution if the incremental regret $\mathbb{E}\left(r^{(\rho)}(x^*) - r^{(\rho)}(x^\mu_t)\right)$ goes to zero as $t \rightarrow \infty$. 

Theorem~\ref{thm:upper_smv} below shows that the expected payoff of DPDS converges to the expected payoff of the optimal solution $x^*$. More precisely, it characterizes the rate of convergence and the regret growth rate of DPDS. 

\begin{theorem}
\label{thm:upper_smv}
Let $x_{t+1}^{\text{\tiny{DPDS}}}$ denote the bid of DPDS policy for day $t+1$. Let DPDS parameter choice $\alpha_t = \max(\lceil \alpha t^\gamma \rceil,2)$ with $\gamma \geq1/2$ and $\alpha>0$,  and let \ref{assumption:iid}, \ref{assumption:bounded}, and \ref{assumption:lipschitz} hold. Then, for $t\geq 2$,
\begin{displaymath}
\mathbb{E}(r^{(\rho)}(x^*)-r^{(\rho)}(x_{t+1}^{\text{\tiny{DPDS}}})) \leq C_1 \sqrt{\log{t}/t}+ C_2 t^{-1/2}
\end{displaymath}
and for $T>1$,
\begin{displaymath}
 \mathcal{R}_T^{\text{\tiny{DPDS}}}(f) \leq C\sqrt{T\log{T}},
\end{displaymath}
where $C=2(C_1+C_2)$ and $C_1$ and $C_2$ are positive constants which depend on the values of $K$, $L$, $p$, $B$, $u$, $l$, $\rho$, $\alpha$, and $\gamma$.
\end{theorem}
\begin{proof}
See Appendix~\ref{appendix:upper_proof}.
\end{proof}

Theorem~\ref{thm:upper_smv} is proved by showing that the expected payoff of $x^*_{t+1} = \argmax_{x \in \mathcal{F}_t} r^{  (\rho) } (x)$ converges to that of $x^*$ due to Lipschitz continuity, and the expected payoff of $x_{t+1}^{\text{\tiny{DPDS}}}$ converges to that of $x^*_{t+1}$ via the use of McDiarmid's inequality.

Theorem~\ref{thm:lower} shows that the regret of any policy is lower bounded by $\Omega(\sqrt{T})$. This result implies that the convergence rate of the expected payoff for any policy cannot be faster than $\Omega(1/\sqrt{t})$ because, otherwise, the regret growth would have been slower than $\Omega(\sqrt{T})$. Hence, DPDS achieves the order-optimal convergence as well as the slowest possible regret growth rate up to a logarithmic factor. 

\begin{theorem}
\label{thm:lower}
Consider the case where $K=1$, $B=1$, $\rho=0$. For any bidding policy $\mu$, there exists a distribution $f$ satisfying assumptions \ref{assumption:iid}, \ref{assumption:bounded}, and \ref{assumption:lipschitz} such that 
\begin{displaymath}
R_T^\mu(f) \geq \frac{1}{16\sqrt{5}}\sqrt{T}.
\end{displaymath}
\end{theorem}
\begin{proof}
See Appendix~\ref{appendix:lower_proof}.
\end{proof}

The proof of Theorem~\ref{thm:lower} is derived by showing that, every time the bid is cleared, an incremental regret greater than $T^{-1/2}/(4\sqrt{5})$ is incurred under a distribution; otherwise, it is incurred under another distribution. However, to distinguish
between these two distributions, one needs $\Omega(T)$ samples which results in a regret lower bound given in Theorem~\ref{thm:lower}.

\section{Empirical Study}

\subsection{Setup and Data}

For the empirical study, we consider virtual bids on zonal nodes for two different ISOs: NYISO and PJM. We use historical DA and RT price data from the beginning of 2006 until the end of 2016 of NYISO and PJM zones. This data set is available for all 11 zones of NYISO and for 19 zones of PJM. Since the price varies in time and location, there are $N \times 24$ different trading options every day where $N=11$ for NYISO and $N=19$ for PJM. The prices are per unit (MWh) prices. We consider virtual demand and virtual supply bids simultaneously for all options by using the model presented in Sec.~\ref{sec:virtual_trading_model}. This model requires the knowledge of an upper bound $u_\lambda$ and a lower bound $l_\lambda$ for DA price. We choose $u_\lambda$ and $l_\lambda$ accordingly for each ISO by looking at the range of the historical DA prices in that markets. We set $u_\lambda=1000$ and $l_\lambda=0$ for NYISO; and $u_\lambda=1050$ and $l_\lambda=-30$ for PJM. Consequently, total number of options is $K = 2 \times N \times 24$. 

The DA market for day $t$ closes early in the morning on day $t-1$ for both NYISO and PJM. Hence, all of the RT prices of day $t-1$ cannot be observed before the bid submission for day $t$. Therefore, the most recent observation used for any algorithm was from day $t-2$ to determine the bid for day $t$. 

\subsection{Benchmark Methods}

We compare DPDS with three algorithms. One is UCBID-GR inspired by UCBID \cite{Weedetal:16}. On each day, UCBID-GR sorts all trading options according to their profitabilities, \textit{i.e.}, their historical average DA-RT price spreads. Then, starting from the most profitable option, it sets the bid for an option equal to its historical average RT price\footnote{The bid is set to zero if historical average RT price is negative because bidding less than or equal to zero implies not bidding on that option.} until there isn't any sufficient budget left. 

The second algorithm is a variant of Kiefer-Wolfowitz stochastic approximation method, herein referred to as SA. SA approximates the gradient of the payoff function by using the current observation and updates the bid of each $k$ as follows; 
\begin{displaymath}
x_{t+1,k} = x_{t,k} + a_t (\pi_{t-1,k}-\lambda_{t-1,k}) \hat{\nabla}_{t,k}
\end{displaymath}
where
\begin{displaymath}
\hat{\nabla}_{t,k} = \frac{\mathds{1}{\{ x_{t,k}+c_t \geq \lambda_{t-1,k}\}}-\mathds{1}{\{ x_{t,k}-c_{t,k} \geq \lambda_{t-1,k}\}} }{c_t}.
\end{displaymath}
Then, $x_{t+1}$ is projected to the feasible set $\mathcal{F}$. The step size $a_t$ and $c_t$ of SA were determined by searching for values that provide relatively better payoff and were set as $20000/(t-1)$ and $2000/(t-1)^{0.25}$, respectively.

The last algorithm is SVM-GR, which is inspired by the use of support vector machines (SVM) by Tang et al.\cite{Tangetal:17} to determine if a demand or a supply bid is profitable for an option, \textit{i.e.}, if the price spread is positive or negative. Due to possible correlation of a particular option's price spread on any given day with the price spreads of that and also of other options that are observed recently, for day $t$, the input of SVM for each option is set as the price spreads of all options from day $t-7$ to day $t-2$. To test SVM-GR algorithm at a particular year, for each option, the data from the previous year is used to train SVM and to determine the average profit, \textit{i.e.}, average price spread, and the bid level that will be accepted with 95\% confidence in the event that a demand or a supply bid is profitable. For the test year, on each day, SVM-GR first determines if a demand or a supply bid is profitable for each option. Then, SVM-GR sorts all options according to their average profits, and, starting from the most profitable option, it sets the bid of an option equal to the bid level with 95\% confidence of acceptance until there isn't any sufficient budget left.

The DPDS algorithm was tested for $\rho$ values of 0 and 0.002 to evaluate the performance under a sum of mean-variance objective instance as well as for $\rho=0$ (the risk-neutral objective). To differentiate between these two different choices of $\rho$, let DPDS($\rho$) denote the DPDS algorithm with associated $\rho$ value. The DPDS algorithm parameter $\alpha_t$ was set to be $t-1$. 

\subsection{Empirical Results/Analysis}

\begin{figure}[!t]
\centering
\includegraphics[width=1\columnwidth]{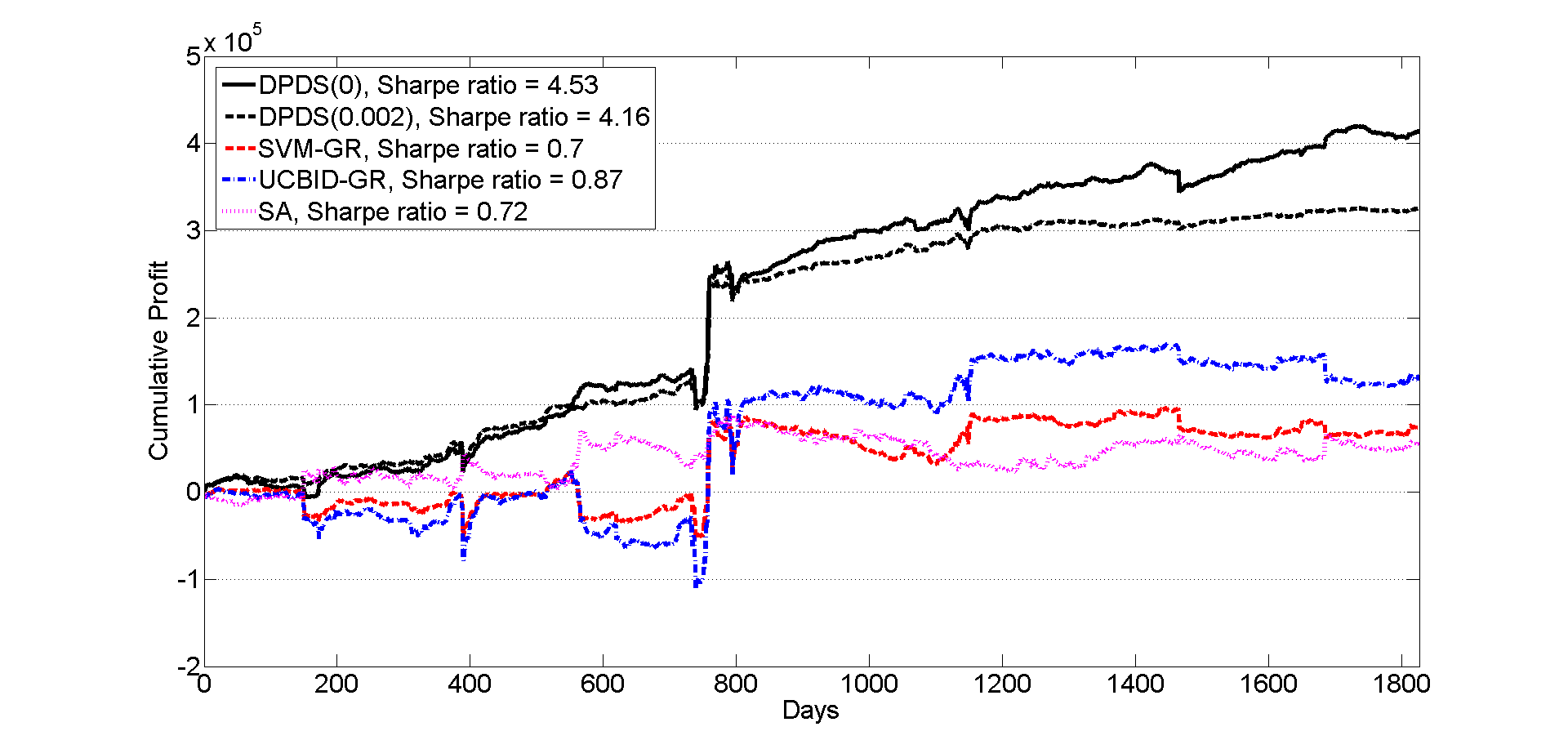}
\caption{Cumulative profit trajectory from 2012 to 2016 in NYISO  for B=\$250,000 after an initial training with 2011 data.}
\label{fig:trajectory}
\end{figure}

\begin{figure*}[!t]
\centering
\subfloat[Annual profit versus year]{\includegraphics[width=1\columnwidth]{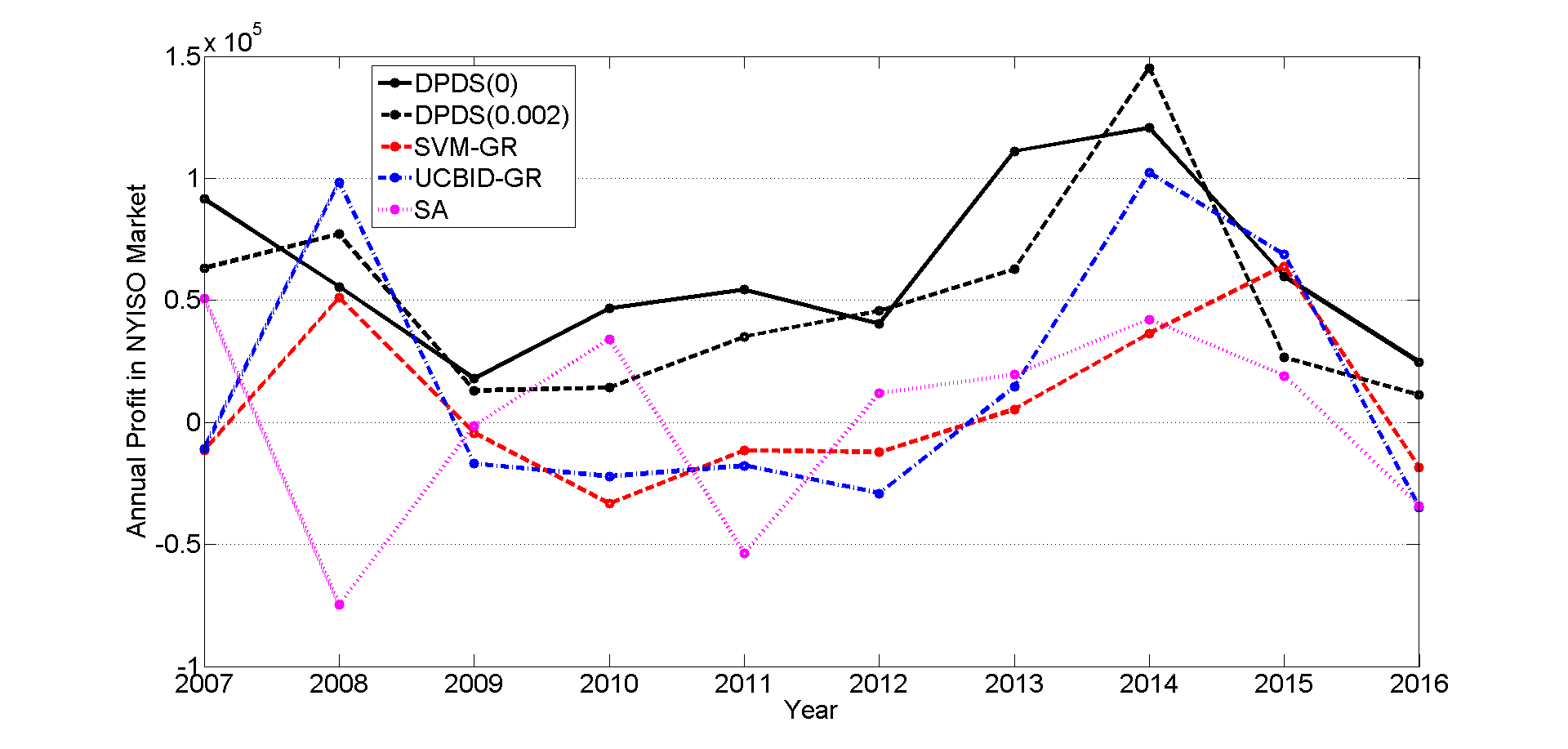}%
\label{fig:nyiso_10years_profit}}
\hfil
\subfloat[Annual Sharpe ratio versus year]{\includegraphics[width=1\columnwidth]{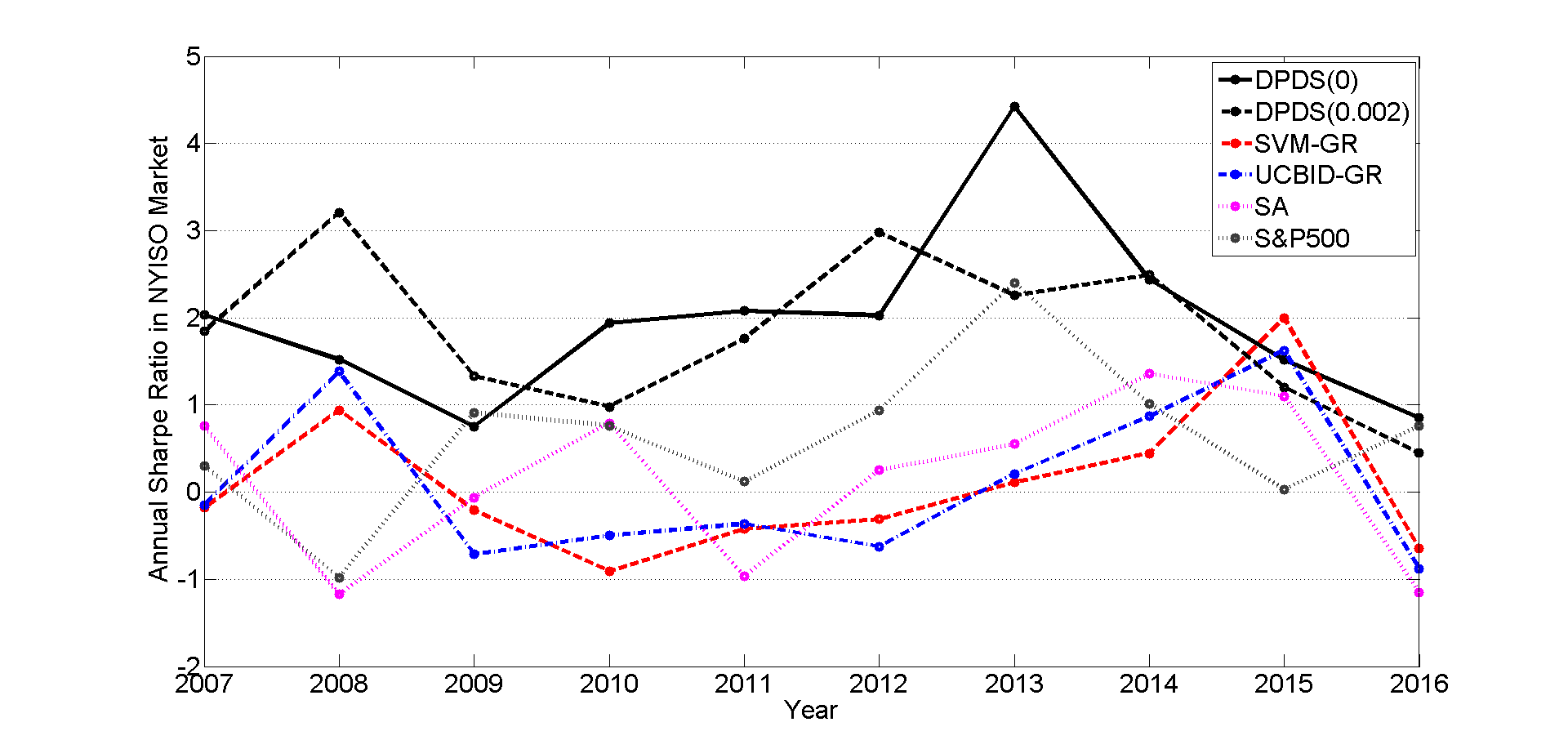}%
\label{fig:nysio_10years_sharpe}}
\caption{Annual performance in NYISO for $B=\$250,000$ (For each year, an initial training with previous year's data was performed.)}
\label{fig:nyiso_10years}
\end{figure*}

\begin{figure*}[!t]
\centering
\subfloat[Annual profit versus year]{\includegraphics[width=1\columnwidth]{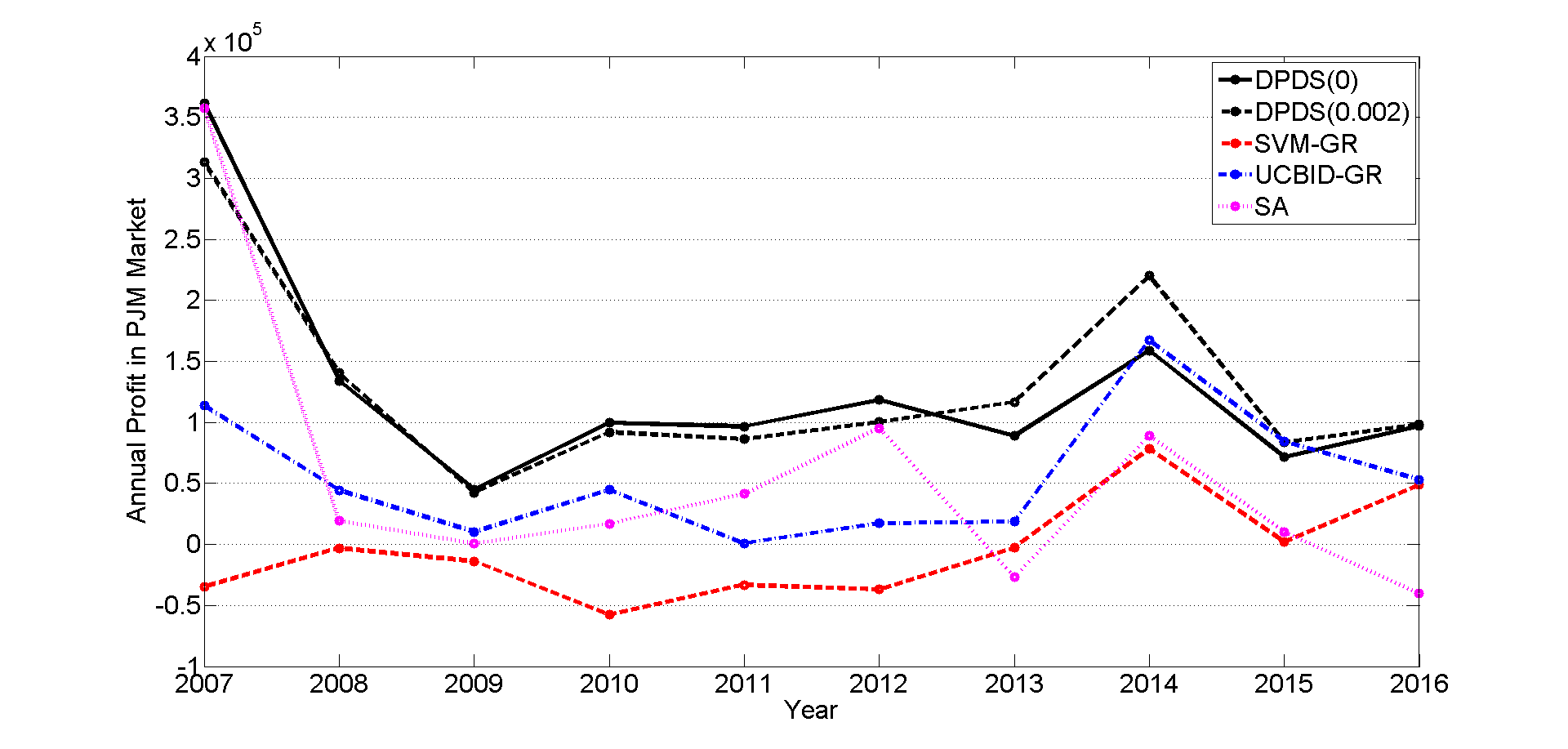}%
\label{fig_pjm_10years_profit}}
\hfil
\subfloat[Annual Sharpe ratio versus year]{\includegraphics[width=1\columnwidth]{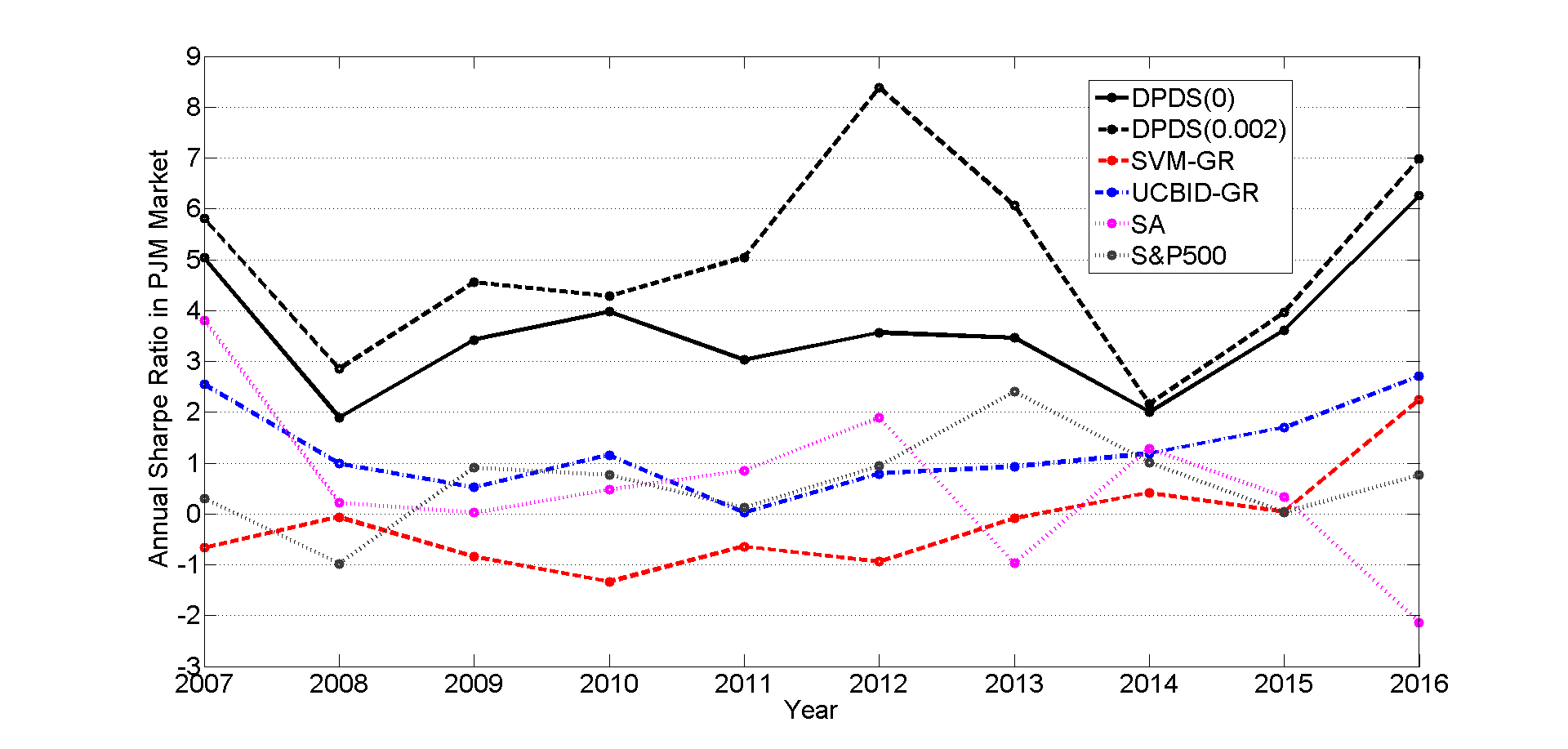}%
\label{fig_pjm_10years_sharpe}}
\caption{Annual performance in PJM for $B=\$250,000$ (For each year, an initial training with previous year's data was performed.)}
\label{fig:pjm_10years}
\end{figure*}

\begin{figure*}[!t]
\centering
\subfloat[2016 profit versus budget level]{\includegraphics[width=1\columnwidth]{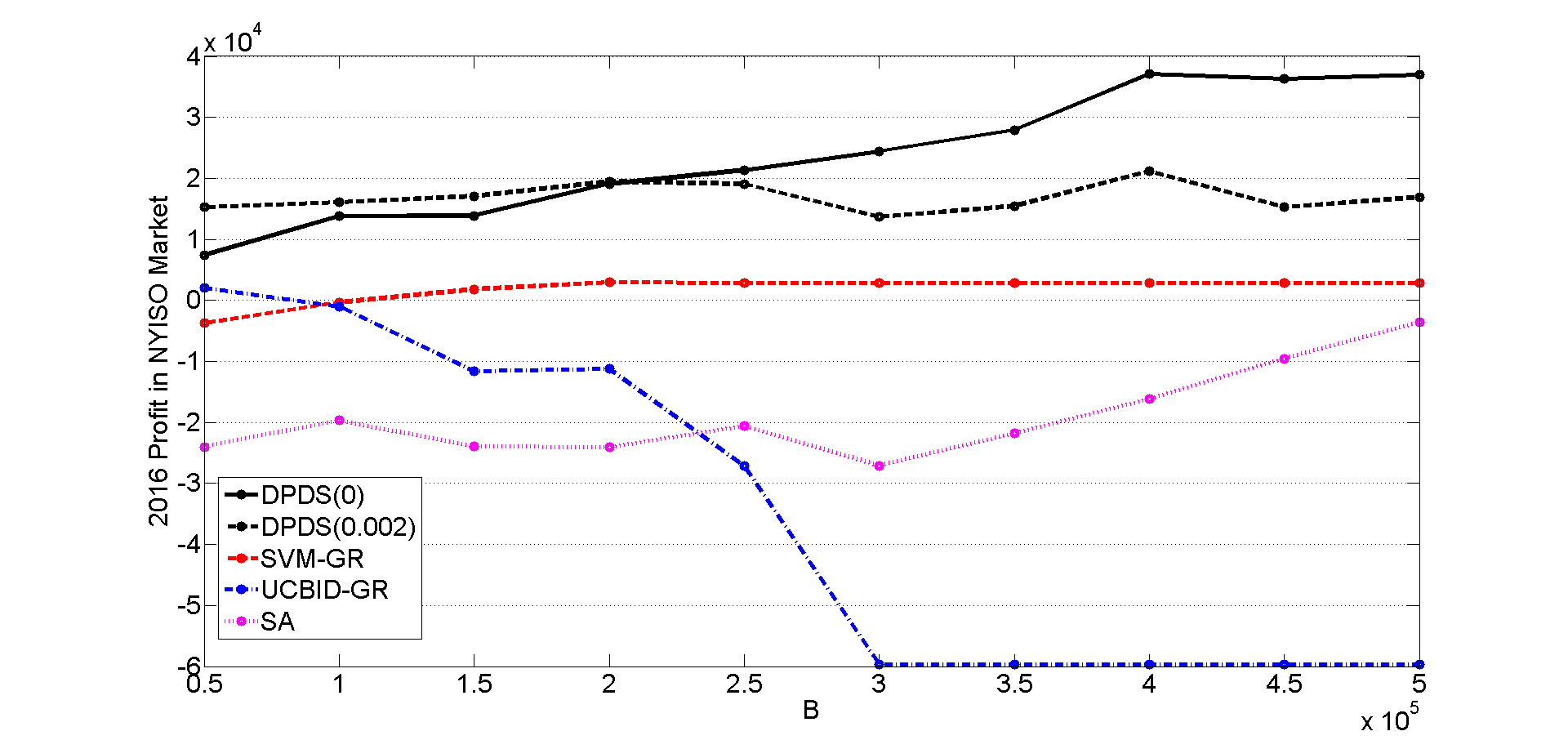}%
\label{fig:nyiso_budget_profit}}
\hfil
\subfloat[2016 Sharpe ratio versus budget level]{\includegraphics[width=1\columnwidth]{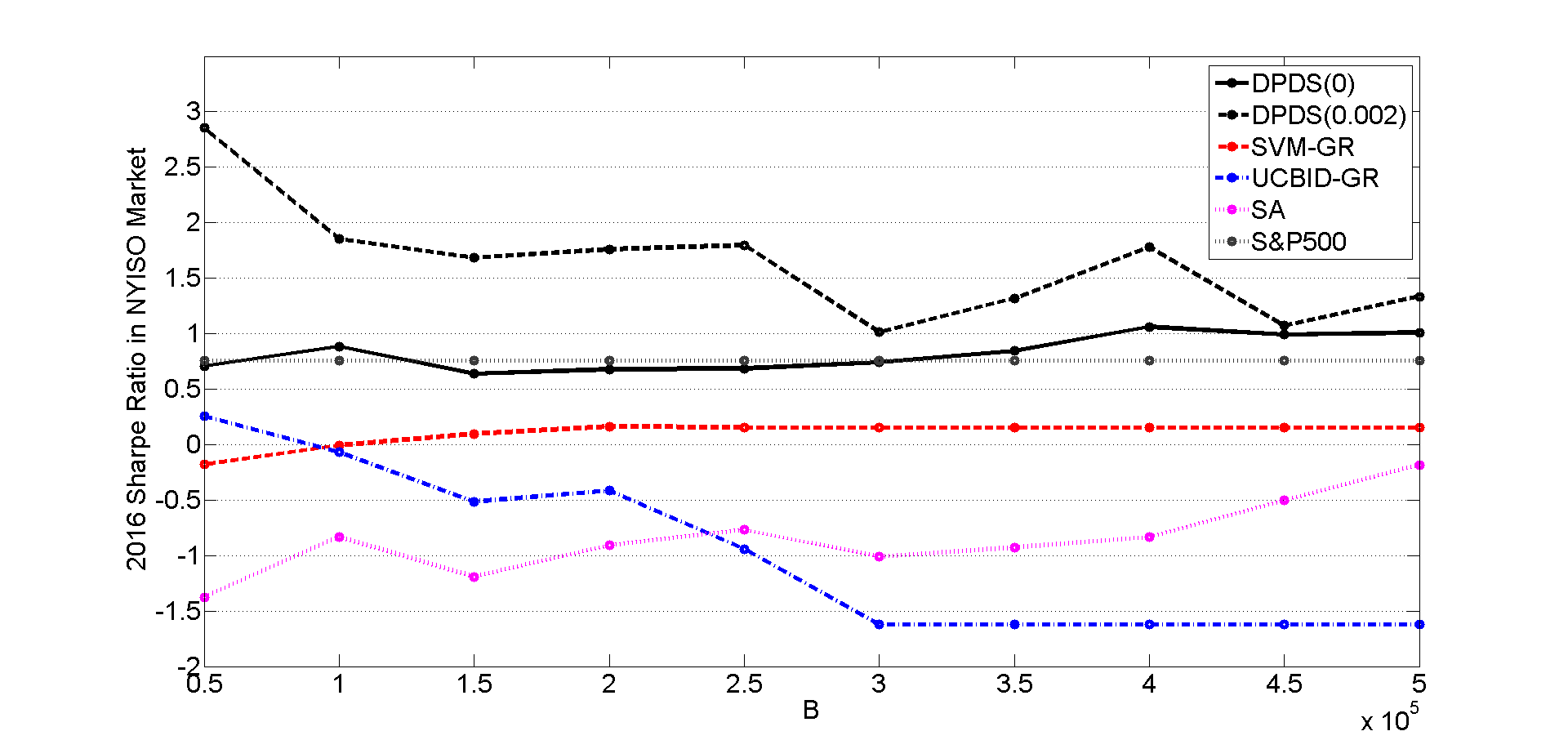}%
\label{fig:nysio_budget_sharpe}}
\caption{2016 Performance in NYISO under different budget levels after an initial training with 2014 and 2015 data}
\label{fig:nyiso_budget}
\end{figure*}

For each algorithm, the trajectory of cumulative profit that was obtained in NYISO market with a daily budget of B=\$250,000 from the beginning of 2012 until the end of 2016 is given in Fig.~\ref{fig:trajectory}. Since the data of 2011 was required to train SVM-GR, other algorithms were also trained starting from the beginning of 2011. First, we observed that DPDS significantly outperformed other algorithms in terms of Sharpe ratio\footnote{In this paper, Sharpe ratio is calculated as $\sqrt{T}\frac{\bar{r}_T}{\sqrt{\frac{1}{T-1}\sum_{t=1}^T(r_t-\bar{r}_T)^2}}$ where $\bar{r}_T = \frac{1}{T}\sum_{t=1}^Tr_t$, $T$ is the number of trading days during the time period under consideration, and $r_t$ is the percentage return of day $t$, which is equal to the profit of day $t$ for virtual trading with fixed daily budget.}, including the S\&P 500 Sharpe ratio\footnote{To calculate this, S\&P 500 adjusted closing price data for the time period under consideration is used. This data is obtained from Yahoo finance.} of 2.10 for the same period. See the legend of Fig.~\ref{fig:trajectory}. This showed the significant value of participating in virtual trading in terms of profitability and risk trade-off.

The cumulative profit of DPFS, as shown in Fig.~\ref{fig:trajectory}, outperformed all other algorithms with DPDS(0), which generated the highest profit. Comparing DPDS(0) and DPDS(0.002) with the latter taking into account the variance of the return, we observed from Fig.~\ref{fig:trajectory} that DPDS(0.002) generated a smoother return trajectory by avoiding more risky bids and generating less profit. We observed that, even though other algorithms were profitable; the increase in their cumulative profits was not consistent. Particularly, for UCBID-GR and SVM-GR, most of their profit resulted from a jump occurred in January 2014 due to a polar vortex \cite{NYISO:15}, which didn't affect SA because of SA's incremental bid update via a local search.

To gain insights from the performance of these algorithms on a yearly basis, annual performances for 10 consecutive years in NYISO market and in PJM market are provided in Fig.~\ref{fig:nyiso_10years} and in Fig.~\ref{fig:pjm_10years}, respectively. To evaluate the performance of a given year, SVM-GR used the data from the previous year for training. Hence, all other algorithms were trained for each year starting from the beginning of the previous year. Fig.~\ref{fig:nyiso_10years}(a) illustrates the total profit that is obtained each year in NYISO. We observed that DPDS outperformed all other algorithms almost every year and consistently achieved a positive profit each year for both $\rho$ values; whereas, all other algorithms incurred losses frequently. Due to the increasing trend in profits from 2009 to 2014, we couldn't conclude that there was a decrease in profits over the years as a result of price convergence despite the decrease in the last two years. In NYISO, 2016 seemed to be the worst year in terms of profitability in general. Annual Sharpe ratios of all algorithms along with that of S\&P 500 are illustrated in Fig.~\ref{fig:nyiso_10years}(b) for NYISO. We observed that DPDS outperformed other algorithms and S\&P 500 also in terms of Sharpe ratio.

Similarly, total profit and Sharpe ratios that were achieved each year in the PJM market are provided in Fig.~\ref{fig:pjm_10years}(a) and Fig.~\ref{fig:pjm_10years}(b), respectively. In PJM, we observed that the trends in terms of both profit and Sharpe ratio were similar to the ones observed in NYISO. In general, we observed that the profit margins of all algorithms except SVM-GR were much higher in PJM compared with NYISO. Similar to the case in NYISO, in PJM, DPDS achieved higher Sharpe ratios than any other algorithm and than S\&P 500. However, in PJM, the performance gap between DPDS and others was much more significant. Furthermore, the Sharpe ratios were in general higher for all algorithms (except SVM-GR) in PJM compared with NYISO counterparts. In PJM, especially DPDS exhibited very high Sharpe ratios, \textit{i.e.}, between 2 and 9, which were consistently higher for $\rho=0.002$ (around 5 on average) compared with $\rho=0$ (around 3.6 on average). 

To illustrate how algorithms performed under different budget constraints, we examined the NYISO market in 2016, the year with the lowest levels of profit and Sharpe ratio (see Fig.~\ref{fig:nyiso_10years}). Total profit and Sharpe ratio for this period under different budget levels are illustrated in Fig~\ref{fig:nyiso_budget}(a) and Fig~\ref{fig:nyiso_budget}(b), respectively. Here, all algorithms were trained initially with the data from the previous two years rather than only previous year. When we increased the data used for initial training to two years, we observed that algorithms performed significantly better in terms of both profit and Sharpe ratio in general. We observed that DPDS outperformed other algorithms at all budget levels, and profit of DPDS(0) increased with increasing budget; whereas the profit of DPDS(0.002) stayed in the same range without an increasing trend. This was reasonable because DPDS(0) optimized profit and should exhibit a profit increase for higher budgets; whereas DPDS(0.002) optimized a linear combination of profit and variance term, which did not indicate a profit increase. SVM-GR also illustrated an increasing trend in profit, but this trend was much smaller compared with the trend of DPDS(0). For both SA and UCBID-GR, big losses were observed almost at all budget levels. In Fig.~\ref{fig:nyiso_budget}(b), we observed that DPDS achieved higher Sharpe ratios than other algorithms for both $\rho$ values, and the Sharpe ratio of DPDS(0) stayed around the Sharpe ratio of S\&P 500; whereas DPDS(0.002) achieved higher Sharpe ratio than DPDS(0) consistently. So, even though the profit levels of 2016 were not as high as the ones that were obtained in previous years, there were bidding strategies that achieved better Sharpe ratio than that of S\&P 500.

\section{Conclusion}

In many wholesale electricity markets, virtual trading is allowed across different locations in the grid and across different hours of the day. For example, PJM allows virtual demand and virtual supply bids on more than 1000 locations. Hence, rather than considering to bid on a specific location (or a specific zone where the price is averaged over many locations), we show that a well-designed online learning algorithm can assess many virtual bidding options simultaneously and select the most profitable ones among them efficiently. Despite that the objective function involved being non-convex and the ERM problem being NP-hard, the proposed polynomial-time algorithm provides an approximate solution that converges to the optimal solution with the fastest rate of convergence. 

More significant, perhaps, is that the proposed algorithm, when tested with multi-year historical data, showed better Sharpe ratio against competitors, including the S\&P 500 index. Such historical data, obviously, do not confirm with the assumption made for the convergence result. This suggests a level of robustness of the proposed algorithm.

There are several directions that the proposed approach can be generalized. The algorithm presented here optimizes the bid values (willingness to pay) for options but not the quantities (number of MWhs). Even though the problem formulation allows optimization of multiple copies of the same option as separate options, this is not efficient in terms of computational complexity. An extension to include quantity as a decision variable should further improve the performance. It would be also interesting to study other risk-averse objectives. For example, including the bid values as well as bid quantities as decision variables to the risk-constrained problem formulation in \cite{Lietal:15} can be considered.

\section*{Acknowledgment}
The authors would like to thank Professor Robert Kleinberg for the insightful discussion.

\ifCLASSOPTIONcaptionsoff
  \newpage
\fi

\appendices

\section{NP-Hardness of the ERM Approach}
\label{appendix:np_hardness}

Due to the piece-wise constant structure, choosing $x_k =\lambda^{(t,k)}_{i}$ for some $i$ contributes the same amount to the overall payoff as choosing any $x_k \in \left[\lambda^{(t,k)}_{i},\lambda^{(t,k)}_{i+1}\right)$. However, choosing $x_k =\lambda^{(t,k)}_{i}$ utilizes a smaller portion of the budget. Hence, an optimal solution to (\ref{eq:avg_opt}) can be obtained by solving the following integer linear program:

\begin{align}
\label{P:ILP}
& \underset{\{z_k\}_{k=1}^K}{\text{maximize}}
& &\sum_{k=1}^K \left(r^{(t,k)}\right)^\intercal z_k \notag \\
& \text{subject to}
& & \sum_{k=1}^K \left(\lambda^{(t,k)}\right)^\intercal z_k \leq B, \\
&&& \|z_{k}\|_1 \leq 1, \qquad  \forall k,\notag \\
&&& z_{k,i} \in \{0,1\}, \qquad \forall (k,i). \notag
\end{align}
where the bid value $x_k = \left(\lambda^{(t,k)}\right)^\intercal z_k$ for node $k$. 

Observe that (\ref{P:ILP}) is a multiple choice knapsack problem (MCKP), a generalization of the 0-1 knapsack. The MCKP problem in (\ref{P:ILP}), unfortunately, is NP-hard\cite{Kellereretal:04}. Had we a polynomial-time algorithm that finds an optimal solution $x \in \mathcal{F}$ to (\ref{eq:avg_opt}), we would have obtained the solution of (\ref{P:ILP}) in polynomial-time by setting $z_{k,i}=1$ where $i= \max_{i: \lambda_{i}^{(t,k)} \leq x_k} i$ for each $k$. By contradiction, the ERM problem (\ref{eq:avg_opt}) is also NP-hard. 

\section{The Bellman Equation and Computational Complexity of DPDS}
\label{appendix:comp_complexity}
Recall that $V_n(b)$ is defined as the maximum payoff one can collect in state $b$ over the remaining $n$ stages. Assuming that $V_0(b)=0$ for any $b$, the Bellman equation can be written as
\begin{equation}
\label{eq:recursion}
V_{K-k+1}(b)= \max_{x_k \in A_{t,b}} \left(\bar{r}_{t,k}(x_k)+V_{K-k}(b-x_k)\right),
\end{equation}
which can be solved via backward induction starting from $k=K$ and proceeding toward $k=1$. For each $k$,  $V_{K-k+1}(b)$ is calculated for all $b \in \mathcal{D}_t$. Since the computation of $V_{K-k+1}(b)$ requires at most $\alpha_t+1$ comparison for any fixed value of $k \in \{1,...,K\}$ and $b \in \mathcal{D}_t$, it has a computational complexity on the order of $K\alpha_t^2$ given the average payoff values $\bar{r}_{t,k}(x_k)$ for all $x_k \in \mathcal{D}_t$ and $k \in \{1,...,K\}$. For each $k \in \{1,...,K\}$, computation of $\bar{r}_{t,k}(x_n)$ for all $x_k \in \mathcal{D}_t$ introduces an additional computational complexity of at most on the order of $t$ which can be achieved by updating $\left(\lambda^{(t,k)},r^{(t,k)}\right)$ from $\left(\lambda^{(t-1,k)},r^{(t-1,k)}\right)$ recursively by using observations  $\lambda_t$ and $\pi_t$ (see the algorithm pseudo-code provided in Fig.~\ref{alg:dpds}). Hence, total computational complexity of DPDS is $O(K\max(t,\alpha_t^2))$ at each day $t$.


\begin{figure}
	\begin{boxedalgorithmic}
		\State \textbf{Initialization:} Set $x_1=0$ and $\left(\lambda^{(k)},r^{(k)},v^{(k)}\right)=(0,0)$ $\forall$ $k \in \{1,...,K\}$; 
		
		\For{$t=1$ to $T$}
		\State Bid $x_t$;
		\State At the end of period $t$, observe $(\lambda_t,\pi_t)$; 
		\For{$k=1$ to $K$}
		\State Set $i_{k} = \max_{i:\lambda^{(k)}_{i}<\lambda_{t,k}}i$; 
		\State $\lambda^{(k)} = \left[\lambda^{(k)}_{1:i_{k}},\lambda_{t,k},\lambda^{(k)}_{i_{k}+1:t}\right]^\intercal$;
		\State  $r^{(k)} = \left[\frac{t-1}{t}r^{(k)}_{1:i_{k}},\frac{t-1}{t}r^{(k)}_{i_{k}:t}+\frac{1}{t}(\pi_{t,k}-\lambda_{t,k})\right]^\intercal$;
		\State  $v^{(k)} = \left[\frac{t-1}{t}v^{(k)}_{1:i_{k}},\frac{t-1}{t}v^{(k)}_{i_{k}:t}+\frac{1}{t}(\pi_{t,k}-\lambda_{t,k})^2\right]^\intercal$;
		\EndFor
		\State Set $V_0(jB/\alpha_t) =0$ $\forall$ $j \in \{0,1,...,\alpha_t\}$;
		\State Set $V_n(0) = 0$ $\forall$ $n \in \{1,...,K\}$; 
		\State Set $w_n(0)=0$ $\forall$ $n \in \{1,...,K\}$; 
		\For{$n=1$ to $K$}
		\State $l = 2$, $d=0$, and $j' = \alpha_t$;
		\For{$j=1$ to $\alpha_t$}
		\While{$d = 0$}
		\If{$\lambda^{(n)}_{l} > jB/\alpha_t$}
		\State $\hat{r}_{t,n}\left(jB/\alpha_t\right) = r^{(n)}_{l-1} + \rho \frac{t}{t-1}[(r^{(n)}_{l-1})^2 - v^{(n)}_{l-1}]$;
		\State break;
		\Else
		\If{$l=t+1$}
		\State $\hat{r}_{t,n}\left(jB/\alpha_t\right) = r^{(n)}_{l} + \rho \frac{t}{t-1}[(r^{(n)}_{l})^2 - v^{(n)}_{l}]$;
		\State $d = 1$ and $j' = j$;
		\State break;
		\Else
		\State $l = l+1$;
		\EndIf
		\EndIf
		\EndWhile
		\State $V_{n}(jB/\alpha_t) = V_{n-1}(jB/\alpha_t)$;
		\State $w_{n}(jB/\alpha_t) = 0$; 
		\For{$i=1$ to $\min\{j,j'\}$}
		\If{$V_{n}(jB/\alpha_t) < V_{n-1}((j-i)B/\alpha_t)+\hat{r}_{t,n}(iB/\alpha_t) $}
		\State $V_{n}(jB/\alpha_t) = V_{n-1}((j-i)B/\alpha_t)+\hat{r}_{t,n}(iB/\alpha_t)$; 
		\State $w_{n}(jB/\alpha_t) = iB/\alpha_t$;
		\EndIf 			
		\EndFor		
		\EndFor
		\EndFor
		\State $B_r=B$;
		\For{$k=K$ to $1$}
		\State $x_{t+1,k} = w_{k}(B_r)$;
		\State $B_r = B_r - x_{t+1,k}$;	
		\EndFor
		\EndFor
	\end{boxedalgorithmic}
	\caption{DPDS algorithm pseudo-code. Here, for a vector $y$, $y_{m:n}=(y_m,y_{m+1},...,y_n)$ denotes the sequence of entries from $m$ to $n$.} 
	\label{alg:dpds}
\end{figure}

\section{Proof of Theorem \ref{thm:upper_smv}}
\label{appendix:upper_proof}

Recall that $x^* = \argmax_{x \in \mathcal{F}} r^{(\rho)}(x)$ and let $x_{t+1}^* = \arg\max_{x \in \mathcal{F}_t} r^{(\rho)}(x)$. Hence, for any $x' \in \mathcal{F}_t$,
\begin{displaymath}
r^{(\rho)}(x^*)-r^{(\rho)}(x_{t+1}^*) \leq r^{(\rho)}(x^*)-r^{(\rho)}(x').
\end{displaymath} 
We take $x'_{k} = \lfloor x^*_k/(B/\alpha_t)\rfloor (B/\alpha_t)$ for all $k \in \{1,...,K\}$, where $\lfloor x^*_k/(B/\alpha_t) \rfloor$ denotes the largest integer smaller or equal to $x^*_k/(B/\alpha_t)$, so that $x' \in \mathcal{F}_t$ and $|x'_{k}-x^*_k| \leq B/\alpha_t$ for all $k \in \{1,...,K\}$. Then, due to Lipschitz continuity of $r^{(\rho)}(.)$ on $\mathcal{F}$ with p-norm and constant $L$,
\begin{equation}
\label{eq:part1}
r^{(\rho)}(x^*)-r^{(\rho)}(x_{t+1}^*)\leq LK^{1/p} B/ \alpha_t.
\end{equation}

Since the payoff obtained at each period $t$ from bidding on a node $k \in \{1,...,K\}$ is in $[l,u]$ and $r^{(\rho)}(.)$ is Lipschitz, $r^{(\rho)}(x^*_{t+1})-r^{(\rho)}(x) \leq c_1$  for any $x \in \mathcal{F}_t$ where $c_1 = \min\left(c_2,LK^{1/p}B\right)$ and $c_2 =K\left((u-l)+\rho(u-l)^2\right)$. Then, for any $\delta_t>0$,
\begin{align*}
r^{(\rho)}(x^*_{t+1})& - r^{(\rho)}(x_{t+1}^{\text{\tiny{DPDS}}})\\
& = \sum_{x \in \mathcal{F}_t} \left(r^{(\rho)}(x^*_{t+1})-r^{(\rho)}(x)\right)\mathds{1}\left\{x_{t+1}^{\text{\tiny{DPDS}}}=x\right\} \\
&  \leq \delta_t \sum_{x \in \mathcal{F}_t:r^{(\rho)}(x^*_{t+1})-r^{(\rho)}(x) \leq \delta_t}\mathds{1}\left\{x_{t+1}^{\text{\tiny{DPDS}}}=x\right\} \\
& \quad + c_1 \sum_{x \in \mathcal{F}_t:r^{(\rho)}(x^*_{t+1})-r^{(\rho)}(x) > \delta_t}\mathds{1}\left\{x_{t+1}^{\text{\tiny{DPDS}}}=x\right\} \\
& \leq \delta_t + c_1 \sum_{x \in \mathcal{F}_t:r^{(\rho)}(x^*_{t+1})-r^{(\rho)}(x) > \delta_t}\mathds{1}\left\{x_{t+1}^{\text{\tiny{DPDS}}}=x\right\}
\end{align*}
where the last inequality is obtained by the fact that at most one of the indicator functions can be equal to one due to the events being disjoint. 

Since DPDS chooses $x \in \mathcal{F}_t$ that maximizes $\bar{r}^{(\rho)}_t(x)= \sum_{k=1}^K\bar{r}^{(\rho)}_{t,k}(x_k)$, $\bar{r}^{(\rho)}_t(x) \geq \bar{r}^{(\rho)}_t(x^*_{t+1})$ has to hold for any $x \in \mathcal{F}_t$ if $x_{t+1}^{\text{\tiny{DPDS}}}=x$. Hence, we can upper bound the last inequality obtained to get
\begin{align*}
& r^{(\rho)}(x^*_{t+1}) -r^{(\rho)}(x_{t+1}^{\text{\tiny{DPDS}}}) \\
& \leq \delta_t + c_1 \sum_{x \in \mathcal{F}_t:r^{(\rho)}(x^*_{t+1})-r^{(\rho)}(x) > \delta_t}\mathds{1}\{\bar{r}^{(\rho)}_t(x) \geq \bar{r}^{(\rho)}_t(x^*_{t+1})\}. 
\end{align*}

In order for $\bar{r}^{(\rho)}_t(x) \geq \bar{r}^{(\rho)}_t(x^*_{t+1})$ to hold for any $x \in F_t$ satisfying $r^{(\rho)}(x^*_{t+1})-r^{(\rho)}(x) > \delta_t$, observe that the event
\begin{displaymath}
\mathcal{E}_1 = \left\{\bar{r}^{(\rho)}_t(x^*_{t+1})+ \delta_t/2 \leq r^{(\rho)}(x^*_{t+1})\right\}
\end{displaymath} 
holds and/or the event
\begin{displaymath}
\mathcal{E}_2 = \left\{ r^{(\rho)}(x)+\delta_t/2 \leq \bar{r}^{(\rho)}_t(x)\right\}
\end{displaymath}
holds. Consequently,
\begin{multline*}
\hspace{-5pt}\mathbb{E}\left(r^{(\rho)}(x^*_{t+1})-r^{(\rho)}(x_{t+1}^{\text{\tiny{DPDS}}})\right) \\
\leq \delta_t + c_1  \sum_{x \in \mathcal{F}_t:r^{(\rho)}(x^*_{t+1})-r^{(\rho)}(x) > \delta_t}  \text{Pr}\left(\mathcal{E}_1 \cup \mathcal{E}_2\right) .
\end{multline*}

Also, observe that, for any fixed $x \in \mathcal{F}$, $\mathbb{E}\left(\bar{r}^{(\rho)}_t(x)\Big|x\right)=r^{(\rho)}(x)$. So, we can use McDiarmid's inequality \cite{McDiarmid:89} to upper bound both $\text{Pr}(\mathcal{E}_1)$ and $\text{Pr}(\mathcal{E}_2)$ if we show that $\bar{r}^{(\rho)}_t(x)$ for fixed $x \in \mathcal{F}_t$ satisfies the bounded differences condition as a function of $\left\{(\lambda_i,\pi_i)\right\}_{i=1}^t \in \Pi^t$ where $\Pi$ denotes the support of the random variable $(\lambda_i,\pi_i)$.

Recall that $\nu_{i,k}$ denotes $(\pi_{i,k}-\lambda_{i,k})$. Define $\bar{r}_{t}^{(-j)}(x) = \sum_{k=1}^K\bar{r}_{t,k}^{(-j)}(x_k)$ where 
\begin{displaymath}
\bar{r}_{t,k}^{(-j)}(x_k) = \frac{1}{t-1}\sum_{i:i\neq j, 1\leq i\leq t} \nu_{i,k}\mathds{1}\{x_k \geq \lambda_{i,k}\},
\end{displaymath} 
and define $\bar{v}_{t}^{(-j)}(x) = \sum_{k=1}^K\bar{v}_{t,k}^{(-j)}(x_k)$ where 
\begin{multline*}
\hspace{-5pt} \bar{v}_{t,k}^{(-j)}(x_k) \\
= \frac{1}{t-1}\sum_{i:i\neq j, 1\leq i\leq t} \left(\nu_{i,k}\mathds{1}\{x_k \geq \lambda_{i,k}\}-\bar{r}_{t,k}^{(-j)}(x_k)\right)^2.
\end{multline*} 
Then, for any $j \in \{1,...,t\}$, we can express $\bar{r}^{(\rho)}_t(x)$ as follows:
\begin{align*}
\bar{r}^{(\rho)}_t(x) & = h^{(-j)}_t(x)+\frac{1}{t}\sum_{k=1}^K\nu_{j,k}\mathds{1}\{x_k \geq \lambda_{j,k}\} \\
& \quad -\frac{\rho}{t}\sum_{k=1}^K \left(\nu_{j,k}\mathds{1}\{x_k \geq \lambda_{j,k}\}-\bar{r}_{t,k}^{(-j)}(x_k)\right)^2
\end{align*}
where 
\begin{displaymath}
h^{(-j)}_t(x)= \frac{t-1}{t} \bar{r}_{t}^{(-j)}(x) -\rho \bar{v}_{t}^{(-j)}(x)
\end{displaymath} 
doesn't depend on $(\lambda_j,\pi_j)$. We also define $\bar{r}^{(\rho,j')}_t(x)$ as
\begin{align*}
\bar{r}^{(\rho,j')}_t(x)& = h^{(-j)}_t(x)+\frac{1}{t}\sum_{k=1}^K\nu_{j',k}\mathds{1}\{x_k \geq \lambda_{j',k}\} \\
& \quad -\frac{\rho}{t}\sum_{k=1}^K \left(\nu_{j',k}\mathds{1}\{x_k \geq \lambda_{j',k}\}-\bar{r}_{t,k}^{(-j)}(x_k)\right)^2.
\end{align*}

Recall that, for any $(\lambda_i,\pi_i) \in \Pi$, $x\in \mathcal{F}$ and $k \in \{1,..,K\}$, $l \leq \nu_{i,k}\mathds{1}\{x_k \geq \lambda_{i,k}\} \leq u$. Therefore, for any $j \in \{1,...,t\}$ and $x\in \mathcal{F}$, $\bar{r}^{(\rho)}_t(x),\bar{r}^{(\rho,j')}_t(x) \in \left[h^{(-j)}_t(x)+K(l-\rho(u-l)^2)/t,h^{(-j)}_t(x)+Ku/t\right]$ for any $\left\{(\lambda_i,\pi_i)\right\}_{i=1}^t,(\lambda_{j'},\pi_{j'})\in \Pi^{t+1}$. Hence, for any $x \in \mathcal{F}$ and $j \in \{1,...,t\}$,
\begin{multline*}
\sup_{\{(\lambda_i,\pi_i)\}_{i=1}^t,(\lambda_{j'},\pi_{j'})\in \Pi^{t+1}} \left|\bar{r}_t^{(\rho)}(x)-\bar{r}_t^{(\rho)(j')}(x)\right| \leq \frac{c_2}{t}. 
\end{multline*}

Since bounded differences condition holds, by McDiarmid's inequality, both $\text{Pr}(\mathcal{E}_1)$ and $\text{Pr}(\mathcal{E}_2)$ are upper bounded by $\exp\left(-t\delta_t^2/\left(2c_2^2\right)\right)$.  Using the fact that the cardinality of the set $\left\{x \in \mathcal{F}_t:r^{(\rho)}(x^*_{t+1})-r^{(\rho)}(x) > \delta_t\right\}$ is upper bounded by $\alpha_t^K+K \leq 2\alpha_t^K$ for $\alpha_t \geq 2$ and $\text{Pr}(\mathcal{E}_1 \cup \mathcal{E}_2) \leq \text{Pr}(\mathcal{E}_1) +\text{Pr}(\mathcal{E}_2)$, we get
\begin{equation}
\label{eq:part2}
\hspace{-5pt} \mathbb{E}\left(r^{(\rho)}(x^*_{t+1})-r^{(\rho)}(x_{t+1}^{\text{\tiny{DPDS}}})\right) 
\leq  \delta_t + 4c_1 \alpha_t^K \exp\left(-\frac{t\delta_t^2}{2c_2^2}\right). 
\end{equation}

By setting $\delta_t =c_2\sqrt{2(\gamma+1)K + 1} \sqrt{\log{t}/t}$ and $\alpha_t =\max(\lceil \alpha t^\gamma \rceil,2)$ with $\gamma\geq 1/2$ and $\alpha>0$, from (\ref{eq:part1}) and (\ref{eq:part2}), we obtain
\begin{align*}
\hspace{5pt} & \mathbb{E}(r^{(\rho)}(x^*)-r^{(\rho)}(x_{t+1}^{\text{\tiny{DPDS}}})) \\
& =  \mathbb{E}(r^{(\rho)}(x^*)-r^{(\rho)}(x_{t+1}^*)) +  \mathbb{E}(r^{(\rho)}(x^*_{t+1})-r^{(\rho)}(x_{t+1}^{\text{\tiny{DPDS}}}))  \\
& \leq LK^{1/p}B/\alpha_t + C_1 \sqrt{\log{t}/t}  + 4c_1\alpha_t^K t^{-(\gamma+1)K-1/2}\\
& \leq  \left(LK^{1/p}B/\alpha+4c_1 \max\left(t^{-K/2},((\alpha+1)/t)^K \right)\right)t^{-1/2}\\
& \quad + C_1 \sqrt{\log{t}/t}\\
& \leq  C_1 \sqrt{\log{t}/t} +C_2 t^{-1/2},
\end{align*}
where $C_1=c_2 \sqrt{2(\gamma+1)K + 1}$ and $C_2=LK^{1/p}B/\alpha+4c_1 \max(1,\alpha^K )$.

For any $T \geq 2$, $\sum_{t=1}^{T-1} 1/\sqrt{t} \leq 2\sqrt{T-1}-1$ and $\sum_{t=1}^{T-1} \sqrt{\log{t}/t} \leq 2\sqrt{(T-1)\log(T-1)}$. Hence, for $T > 2$,
\begin{align*}
\sum_{t=2}^{T-1}& \mathbb{E}\left(r^{(\rho)}(x^*)-r^{(\rho)}(x_{t+1}^{\text{\tiny{DPDS}}})\right)\\
& \leq C_1 \sum_{t=1}^{T-1} \sqrt{\frac{\log{t}}{t}} + C_2\sum_{t=1}^{T-1}\frac{1}{ \sqrt{t}} \\
& \leq  2 C_1 \sqrt{(T-1)\log(T-1)} + C_2\left(2\sqrt{T-1}-1\right) .
\end{align*}
Since $\mathbb{E}\left(r^{(\rho)}(x^*)-r^{(\rho)}(x_{t}^{\text{\tiny{DPDS}}})\right) \leq c_1$, for any $T \geq 1$,
\begin{equation*}
\mathcal{R}_T^{\text{\tiny{DPDS}}}(f) \leq 2C_1 \sqrt{T\log{T}} + 2 C_2 \sqrt{T}
\end{equation*} 
and for any  $T > 1$,
\begin{equation*}
\mathcal{R}_T^{\text{\tiny{DPDS}}}(f) \leq 2 (C_1+C_2) \sqrt{T\log{T}}.
\end{equation*} \hfill \IEEEQED


\section{Proof of Theorem \ref{thm:lower}}
\label{appendix:lower_proof}
Let $\lambda_t$ and $\pi_t$ be independent random variables with distributions 
\begin{displaymath}
f_\lambda(\lambda_t)= \epsilon^{-1}\mathds{1}\{(1-\epsilon)/2 \leq \lambda_t \leq (1+\epsilon)/2\}
\end{displaymath}
and $f_{\pi}(\pi_t) = \text{Bernoulli}(\bar{\pi})$, respectively. Let $f(\lambda_t,\pi_t) =f_\lambda(\lambda_t) f_{\pi}(\pi_t)$ and $\epsilon= T^{-1/2}/2\sqrt{5}$.

Fix any policy $\mu$. Since $\lambda_t$ and $\pi_t$ are independent,
\begin{equation*}
\label{eq:payoff_independent}
r^{(0)}(x) = \mathbb{E}((\bar{\pi}-\lambda_t)\mathds{1}\{x \geq \lambda_t\}|x)
\end{equation*}
and
\begin{multline}
\label{eq:incregret_independent}
r^{(0)}(x^*) - r^{(0)}(x_t^\mu) \\
= \mathbb{E}((\bar{\pi}-\lambda_t)(\mathds{1}\{x^* \geq \lambda_t\} - \mathds{1}\{x_t^\mu \geq \lambda_t\})|x_t^\mu,x^*)
\end{multline} 

Let $f_0$, $f_1$, $f_2$ denote the distribution of $\{\lambda_t, \pi_t\}_{t=1}^T$ and policy $\mu$ under the choice of $\bar{\pi}=1/2$, $\bar{\pi}=1/2-\epsilon$, and $\bar{\pi}=1/2+\epsilon$, respectively. Also, let $\mathbb{E}_i(.)$  and $\mathcal{R}_T^\mu(f_i)$ denote the expectation with respect to the distribution $f_i$ and the regret of policy $\mu$ under  distribution $f_i$, respectively.  

Under distribution $f_1$, observe that $\bar{\pi}-\lambda_t \leq -\epsilon/2 $ for any value of $\lambda_t$ . Therefore, optimal solution under known distribution $x^* \in [0,(1-\epsilon)/2]$ so that $\mathds{1}\{x^* \geq \lambda_t\} = 0$. Then, by (\ref{eq:incregret_independent}), the regret can be expressed as
\begin{align*}
\mathcal{R}_T^\mu(f_1) & = \mathbb{E}_{1}\left(\sum_{t=1}^T -(\bar{\pi}-\lambda_t)\mathds{1}\{x_t^\mu \geq \lambda_t\}\right) \\ 
& \geq \frac{\epsilon}{2}\mathbb{E}_{1}\left(\sum_{t=1}^T\mathds{1}\{x_t^\mu \geq \lambda_t\}\right).
\end{align*}

Similarly, under distribution $f_2$, observe that $\bar{\pi}-\lambda_t \geq \epsilon/2 $ for any value of $\lambda_t$ . Therefore, optimal solution under known distribution $x^* \in [(1+\epsilon)/2, 1]$ so that $\mathds{1}\{x^* \geq \lambda_t\} = 1$. Then, by (\ref{eq:incregret_independent}), the regret becomes
\begin{align*}
\mathcal{R}_T^\mu(f_2) & = \mathbb{E}_{2}\left(\sum_{t=1}^T(\bar{\pi}-\lambda_t)\mathds{1}\{x_t^\mu < \lambda_t\}\right) \\ & \geq \frac{\epsilon}{2}\mathbb{E}_{2}\left(\sum_{t=1}^T\mathds{1}\{x_t^\mu < \lambda_t\}\right).
\end{align*}

For any non-negative bounded function $h$ defined on information history $I_T=\{x_t,\lambda_t,\pi_t\}_{t=1}^T$ such that  $0 \leq h(I_T) \leq M $ for some $M \geq 0$ and for any distributions $p$ and $q$, the difference between the expected value of $h$ under the distributions $p$ and $q$ is bounded by a function of the KL-divergence between these distributions as follows:
\begin{align}
\label{eq:lb_lemma}
\mathbb{E}_q(h(I_T)) -  \mathbb{E}_p & (h(I_T)) \notag \\
& \leq \int_{q(I_T)>p(I_T)} h(I_T)(q(I_T)-p(I_T))dI_T \notag \\
& \leq M \int_{q(I_T)>p(I_T)} q(I_T)-p(I_T)dI_T \notag \\
& = M \frac{1}{2} \int |q(I_T)-p(I_T)| dI_T \notag \\
& \leq M\sqrt{\text{KL}(q||p)/2}.
\end{align}
where $\text{KL}(q||p) = \int q(I_T) \log(q(I_T)/p(I_T)) dI_T$ is the KL-divergence between $q$ and $p$ and the last inequality is due to Pinsker's inequality \cite{Tsybakov:09}, \textit{i.e.},  $V(q,p) \leq \sqrt{\text{KL}(q||p)/2}$ where $V(q,p) = \int |q(I_T)-p(I_T)| dI_T/2$ is the variational distance between $q$ and $p$. The bound given in (\ref{eq:lb_lemma}) is inspired by a similar bound obtained by \cite{Aueretal:02b} in the proof of Lemma A.1 for the case of discrete distribution in the context of non-stochastic multi-armed bandit problem.  

Now, since $\sum_{t=1}^T\mathds{1}\{x_t^\mu \geq \lambda_t\} \leq T$ and  $\sum_{t=1}^T\mathds{1}\{x_t^\mu < \lambda_t\} \leq T$, we use (\ref{eq:lb_lemma}) to obtain
\begin{displaymath}
\mathcal{R}_T^\mu(f_1)  \geq \frac{\epsilon}{2}\mathbb{E}_{0}\left(\sum_{t=1}^T\mathds{1}\{x_t^\mu \geq \lambda_t\}\right) - \frac{\epsilon}{2} T \sqrt{KL(f_0||f_1)/2},
\end{displaymath}
and
\begin{displaymath}
\mathcal{R}_T^\mu(f_2)  \geq \frac{\epsilon}{2}\mathbb{E}_{0}\left(\sum_{t=1}^T\mathds{1}\{x_t^\mu < \lambda_t\}\right) - \frac{\epsilon}{2} T \sqrt{KL(f_0||f_2)/2} .
\end{displaymath} 
Consequently,
\begin{align}
\label{eq:worstcase}
\max_{i \in \{1,2\}} & \mathcal{R}_T^\mu (f_i) \notag \\
& \geq \frac{1}{2}\left(\mathcal{R}_T^\mu(f_1)+\mathcal{R}_T^\mu(f_2)\right) \notag \\
& \geq \frac{\epsilon}{4} \left(T-T\sqrt{\text{KL}(f_0||f_1)/2}-T\sqrt{\text{KL}(f_0||f_2)/2}\right).
\end{align}
For any $i \in \{0,1,2\}$, we can express the distribution of observations in terms of conditional distributions as follows;
\begin{equation*}
\begin{split}
f_i(I_T) & = \prod_{t=1}^T f_i(\pi_t,\lambda_t|x_t^\mu, I_{t-1})f_i(x_t^\mu|I_{t-1})\\
& = \prod_{t=1}^T f_i(\pi_t)f_{\lambda}(\lambda_t) f(x_t^\mu|I_{t-1}),
\end{split}
\end{equation*}
where the second equality is due to the independence of $\lambda_t$ and $\pi_t$ from the past observations $I_{t-1}$, the bid $x_t^\mu$, and from each other. Also, the distribution of $x_t^\mu$ given $I_{t-1}$ does not depend on $i$. Consequently, for $i \in \{1,2\}$,
\begin{equation*}
\begin{split}
\text{KL}(f_0||f_i) & = \int f_0(I_T) \log\left( \prod_{t=1}^T \frac{f_0(\pi_t)}{f_i(\pi_t)}\right)dI_T\\
& = \sum_{t=1}^T \int f_0(I_T) \log\left(  \frac{f_0(\pi_t)}{f_i(\pi_t)}\right) dI_T\\
& = \sum_{t=1}^T  \left(\frac{1}{2}\log\left(\frac{1/2}{1/2+\epsilon}\right)+ \frac{1}{2}\log\left(\frac{1/2}{1/2-\epsilon}\right)\right) \\
& = -(T/2)\log\left(1-4\epsilon^2\right).
\end{split}
\end{equation*}
Then, by (\ref{eq:worstcase}) and by setting $\epsilon= T^{-1/2}/2\sqrt{5}$, we get
\begin{equation*}
\begin{split}
\max_{i \in \{1,2\}} \mathcal{R}_T^\mu(f_i)&  \geq \frac{\epsilon T}{4} \left(1-\sqrt{ -T\log\left(1-4\epsilon^2\right)}\right)\\
& = \frac{\sqrt{T}}{8\sqrt{5}} \left(1-\sqrt{ -T\log\left(1-1/(5T)\right)}\right) \\
& \geq  \frac{\sqrt{T}}{16\sqrt{5}} 
\end{split}
\end{equation*}
where the last inequality follows from the fact that $-\log(1-x)\leq (5/4)x$ for $0 \leq x\leq 1/5$. 

Observe that the magnitude of the derivative of $r^{(0)}(x)$ is equal to $|\bar{\pi}-x|/\epsilon$ for $(1-\epsilon)/2 \leq x \leq (1+\epsilon)/2$ and $0$ otherwise. So, for distributions $f_1$ and $f_2$, $r^{(0)}(x)$ is Lipschitz continuous with Lipschitz constant $L=3/2$ because $|\bar{\pi}-x|/\epsilon \leq 3/2$ for $(1-\epsilon)/2 \leq x \leq (1+\epsilon)/2$. Hence, assumptions \ref{assumption:iid}, \ref{assumption:bounded}, and \ref{assumption:lipschitz} are satisfied for both distributions. \hfill \IEEEQED

 \bibliographystyle{IEEEtran}
\bibliography{reference}

\begin{thebibliography}{10}
\providecommand{\url}[1]{#1}
\csname url@samestyle\endcsname
\providecommand{\newblock}{\relax}
\providecommand{\bibinfo}[2]{#2}
\providecommand{\BIBentrySTDinterwordspacing}{\spaceskip=0pt\relax}
\providecommand{\BIBentryALTinterwordstretchfactor}{4}
\providecommand{\BIBentryALTinterwordspacing}{\spaceskip=\fontdimen2\font plus
\BIBentryALTinterwordstretchfactor\fontdimen3\font minus
  \fontdimen4\font\relax}
\providecommand{\BIBforeignlanguage}[2]{{%
\expandafter\ifx\csname l@#1\endcsname\relax
\typeout{** WARNING: IEEEtran.bst: No hyphenation pattern has been}%
\typeout{** loaded for the language `#1'. Using the pattern for}%
\typeout{** the default language instead.}%
\else
\language=\csname l@#1\endcsname
\fi
#2}}
\providecommand{\BIBdecl}{\relax}
\BIBdecl

\bibitem{Lietal:15}
R.~Li, A.~J. Svoboda, and S.~S. Oren, ``Efficiency impact of convergence
  bidding in the california electricity market,'' \emph{J. Regul. Econ.},
  vol.~48, no.~3, pp. 245--284, Dec. 2015.

\bibitem{PJM:15}
\BIBentryALTinterwordspacing
PJM, ``Virtual transactions in the pjm energy markets,'' PJM, Audubon, PA, USA,
  2015. [Online]. Available:
  \url{http://www.pjm.com/~/media/library/reports-notices/special-reports/20151012-virtual-bid-report.ashx}
\BIBentrySTDinterwordspacing

\bibitem{Parsonsetal:15}
\BIBentryALTinterwordspacing
J.~E. Parsons, C.~Colbert, J.~Larrieu, T.~Martin, and E.~Mastrangelo,
  ``Financial arbitrage and efficient dispatch in wholesale electricity
  markets,'' \emph{MIT Center for Energy and Environmental Policy Research No.
  15-002}, 2015. [Online]. Available: \url{https://ssrn.com/abstract=2574397}
\BIBentrySTDinterwordspacing

\bibitem{Hogan:16}
W.~W. Hogan, ``Virtual bidding and electricity market design,''
  \emph{Electricity J.}, vol.~29, no.~5, pp. 33 -- 47, June 2016.

\bibitem{Tangetal:16}
W.~Tang, R.~Rajagopal, K.~Poolla, and P.~Varaiya, ``Model and data analysis of
  two-settlement electricity market with virtual bidding,'' in \emph{2016 IEEE
  55th Conf. Decision and Control}, 2016, pp. 6645--6650.

\bibitem{Jha&Wolak:15}
\BIBentryALTinterwordspacing
A.~Jha and F.~A. Wolak, ``Testing for market efficiency with transactions
  costs: An application to convergence bidding in wholesale electricity
  markets,'' 2015. [Online]. Available:
  \url{http://web.stanford.edu/group/fwolak/cgi-bin/sites/default/files/CAISO\_VB\_draft\_VNBER\_final.pdf}
\BIBentrySTDinterwordspacing

\bibitem{Saravia:03}
\BIBentryALTinterwordspacing
C.~Saravia, ``Speculative trading and market performance: The effect of
  arbitrageurs on efficiency and market power in the new york electricity
  market,'' \emph{Center for the Study of Energy Markets}, 2003. [Online].
  Available: \url{http://escholarship.org/uc/item/0mx44472}
\BIBentrySTDinterwordspacing

\bibitem{Guleretal:10}
T.~Guler, G.~Gross, E.~Litvinov, and R.~Coutu, ``On the economics of power
  system security in multi-settlement electricity markets,'' \emph{IEEE Trans.
  Power Syst.}, vol.~25, no.~1, pp. 284--295, Feb. 2010.

\bibitem{Wooetal:15}
C.~Woo, J.~Zarnikau, E.~Cutter, S.~Ho, and H.~Leung, ``Virtual bidding, wind
  generation and california's day-ahead electricity forward premium,''
  \emph{Electricity J.}, vol.~28, no.~1, pp. 29 -- 48, Feb. 2015.

\bibitem{Weedetal:16}
J.~Weed, V.~Perchet, and P.~Rigollet, ``Online learning in repeated auctions,''
  in \emph{29th Annu. Conf. Learning Theory}, 2016, pp. 1562--1583.

\bibitem{Kleinberg&Slivkins:10}
R.~Kleinberg and A.~Slivkins, ``Sharp dichotomies for regret minimization in
  metric spaces,'' in \emph{Proc. 21th Annu. ACM-SIAM Symp. Discrete
  Algorithms}, 2010, pp. 827--846.

\bibitem{Matheretal:17}
\BIBentryALTinterwordspacing
J.~Mather, E.~Bitar, and K.~Poolla, ``Virtual bidding: Equilibrium, learning,
  and the wisdom of crowds,'' \emph{IFAC-PapersOnLine}, vol.~50, no.~1, pp. 225
  -- 232, 2017, 20th IFAC World Congress. [Online]. Available:
  \url{http://www.sciencedirect.com/science/article/pii/S2405896317300526}
\BIBentrySTDinterwordspacing

\bibitem{Borensteinetal:08}
S.~Borenstein, J.~Bushnell, C.~R. Knittel, and C.~Wolfram, ``Inefficiencies and
  market power in financial arbitrage: A study of california's electricity
  markets*,'' \emph{J. Ind. Econ.}, vol.~56, no.~2, pp. 347--378, June 2008.

\bibitem{Birgeetal:17}
\BIBentryALTinterwordspacing
J.~Birge, A.~Hortacsu, I.~Mercadal, and M.~Pavlin, ``Limits to arbitrage in
  electricity markets: A case study of miso,'' \emph{MIT Center for Energy and
  Enviromental Policy Research}, 2017. [Online]. Available:
  \url{http://ceepr.mit.edu/publications/reprints/656}
\BIBentrySTDinterwordspacing

\bibitem{Vapnik:92}
V.~Vapnik, ``Principles of risk minimization for learning theory,'' in
  \emph{Advances Neural Inform. Process. Syst. 4}, 1992, pp. 831--838.

\bibitem{Kellereretal:04}
H.~Kellerer, U.~Pferschy, and D.~Pisinger, ``The multiple-choice knapsack
  problem,'' in \emph{Knapsack Problems}.\hskip 1em plus 0.5em minus
  0.4em\relax Berlin, Heidelberg: Springer Berlin Heidelberg, 2004, ch.~11, pp.
  317--347.

\bibitem{Dudzinski&Walukiewicz:87}
K.~Dudziński and S.~Walukiewicz, ``Exact methods for the knapsack problem and
  its generalizations,'' \emph{Eur. J. Oper. Res.}, vol.~28, no.~1, pp. 3 --
  21, 1987.

\bibitem{Markowitz:52}
H.~Markowitz, ``Portfolio selection,'' \emph{J. Finance}, vol.~7, no.~1, pp.
  77--91, Mar. 1952.

\bibitem{Tangetal:17}
\BIBentryALTinterwordspacing
W.~Tang, R.~Rajagopal, K.~Poolla, and P.~Varaiya, ``Impact of virtual bidding
  on financial and economic efficiency of wholesale electricity markets,''
  working paper. [Online]. Available:
  \url{http://www4.ncsu.edu/~wtang8/docs/TaRaPoVa.pdf}
\BIBentrySTDinterwordspacing

\bibitem{NYISO:15}
\BIBentryALTinterwordspacing
D.~B. Patton, P.~LeeVanSchaick, and J.~Chen, ``2014 state of the market report
  for the new york iso markets,'' Tech. Rep., May 2015. [Online]. Available:
  \url{http://www.nyiso.com/public/webdocs/markets_operations/documents/Studies_and_Reports/Reports/Market_Monitoring_Unit_Reports/2014/NYISO2014SOMReport__5-13-2015_Final.pdf}
\BIBentrySTDinterwordspacing

\bibitem{McDiarmid:89}
C.~McDiarmid, ``On the method of bounded differences,'' in \emph{Surveys in
  Combinatorics, 1989: Invited Papers at the Twelfth British Combinatorial
  Conference}.\hskip 1em plus 0.5em minus 0.4em\relax Cambridge: Cambridge
  University Press, 1989, pp. 148--188.

\bibitem{Tsybakov:09}
A.~B. Tsybakov, ``Lower bounds on the minimax risk,'' in \emph{Introduction to
  Nonparametric Estimation}.\hskip 1em plus 0.5em minus 0.4em\relax New York,
  NY: Springer New York, 2009, ch.~2, pp. 77--135.

\bibitem{Aueretal:02b}
P.~Auer, N.~Cesa-Bianchi, Y.~Freund, and R.~E. Schapire, ``The nonstochastic
  multiarmed bandit problem,'' \emph{SIAM J. Comput.}, vol.~32, no.~1, pp.
  48--77, 2002.

\end{thebibliography}

\end{document}